\documentclass[journal]{IEEEtran}
\usepackage{amsthm, amsmath, amsfonts, amssymb, amscd}
\usepackage[dvips]{graphicx}
\usepackage{epsfig}
\usepackage{epstopdf}
\usepackage{graphicx}
\usepackage{subfigure}
\usepackage{enumerate}
\usepackage{color}
\usepackage{cite}
\usepackage{tikz}
\usepackage{lipsum}
\usepackage{tabularx}
\usepackage{textcomp}
\usepackage{stfloats}
\usepackage{balance}

\newtheorem{prop}{Proposition}
\DeclareGraphicsExtensions{.eps}
\newcommand{\fbs}{\mathrm{F}}
\newcommand{\pr}{\mathrm{P}}
\newcommand{\psa}{\mathrm{P}_{\scriptstyle s}^{\scriptstyle(A)}}

\newcommand{\psvonea}{\mathrm{P}_{\scriptstyle s,V_1}^{\scriptstyle(A)}}
\newcommand{\psvtwoa}{\mathrm{P}_{\scriptstyle s,V_2}^{\scriptstyle(A)}}
\newcommand{\psvoneb}{\mathrm{P}_{\scriptstyle s,V_1}^{\scriptstyle(B)}}
\newcommand{\psvtwob}{\mathrm{P}_{\scriptstyle s,V_2}^{\scriptstyle(B)}}
\newcommand{\psvonec}{\mathrm{P}_{\scriptstyle s,V_1}^{\scriptstyle(C)}}
\newcommand{\psvtwoc}{\mathrm{P}_{\scriptstyle s,V_2}^{\scriptstyle(C)}}
\newcommand{\psvoned}{\mathrm{P}_{\scriptstyle s,V_1}^{\scriptstyle(D)}}
\newcommand{\psvtwod}{\mathrm{P}_{\scriptstyle s,V_2}^{\scriptstyle(D)}}
\newcommand{\opa}{\mathtt{OP}_{\scriptstyle A}}

\newcommand{\opaa}{\mathtt{OP}_{\scriptstyle a}}
\newcommand{\opbone}{\mathtt{OP}_{\scriptstyle B}^{\scriptstyle(\alpha > 0.5)}}
\newcommand{\opbtwo}{\mathtt{OP}_{\scriptstyle B}^{\scriptstyle(\alpha \leq 0.5)}}
\newcommand{\opcone}{\mathtt{OP}_{\scriptstyle C}^{\scriptstyle(\alpha > 0.5)}}
\newcommand{\opctwo}{\mathtt{OP}_{\scriptstyle C}^{\scriptstyle(\alpha \leq 0.5)}}
\newcommand{\opdone}{\mathtt{OP}_{\scriptstyle D}^{\scriptstyle(\alpha > 0.5)}}
\newcommand{\opdtwo}{\mathtt{OP}_{\scriptstyle D}^{\scriptstyle(\alpha \leq 0.5)}}
\newcommand{\foneone}{F_{1}^{(1)}}
\newcommand{\fonetwo}{F_{1}^{(2)}}
\newcommand{\ftwoone}{F_{2}^{(1)}}
\newcommand{\ftwotwo}{F_{2}^{(2)}}
\newcommand{\xoneone}{x_{1}^{(1)}}
\newcommand{\xonetwo}{x_{1}^{(2)}}
\newcommand{\xtwoone}{x_{2}^{(1)}}
\newcommand{\xtwotwo}{x_{2}^{(2)}}
\newcommand{\goneone}{\gamma_{1}^{(1)}}
\newcommand{\gonetwo}{\gamma_{1}^{(2)}}
\newcommand{\gtwoone}{\gamma_{2}^{(1)}}
\newcommand{\gtwotwo}{\gamma_{2}^{(2)}}
\newcommand{\psvoneaone}{\mathrm{P}_{\scriptstyle s,V_1}^{\scriptstyle(a,1)}}
\newcommand{\psvoneatwo}{\mathrm{P}_{\scriptstyle s,V_1}^{\scriptstyle(a,2)}}
\newcommand{\psvoneathree}{\mathrm{P}_{\scriptstyle s,V_1}^{\scriptstyle(a,3)}}
\newcommand{\psvtwoaone}{\mathrm{P}_{\scriptstyle s,V_2}^{\scriptstyle(a,1)}}
\newcommand{\psvtwoatwo}{\mathrm{P}_{\scriptstyle s,V_2}^{\scriptstyle(a,2)}}
\newcommand{\psvtwoathree}{\mathrm{P}_{\scriptstyle s,V_2}^{\scriptstyle(a,3)}}
\newcommand{\moneone}{m_{\scriptscriptstyle 1}^{\scriptscriptstyle (1)}}
\newcommand{\monetwo}{m_{\scriptscriptstyle 1}^{\scriptscriptstyle (2)}}
\newcommand{\mtwoone}{m_{\scriptscriptstyle 2}^{\scriptscriptstyle (1)}}
\newcommand{\mtwotwo}{m_{\scriptscriptstyle 2}^{\scriptscriptstyle (2)}}
\newcommand{\omoneone}{\Omega_{\scriptscriptstyle 1}^{\scriptscriptstyle (1)}}
\newcommand{\omonetwo}{\Omega_{\scriptscriptstyle 1}^{\scriptscriptstyle (2)}}
\newcommand{\omtwoone}{\Omega_{\scriptscriptstyle 2}^{\scriptscriptstyle (1)}}
\newcommand{\omtwotwo}{\Omega_{\scriptscriptstyle 2}^{\scriptscriptstyle (2)}}
\newcommand{\monei}{m_{\scriptscriptstyle 1}^{\scriptscriptstyle (i)}}
\newcommand{\mtwoi}{m_{\scriptscriptstyle 2}^{\scriptscriptstyle (i)}}
\newcommand{\omonei}{\Omega_{\scriptscriptstyle 1}^{\scriptscriptstyle (i)}}
\newcommand{\omtwoi}{\Omega_{\scriptscriptstyle 2}^{\scriptscriptstyle (i)}}
\begin{document}

\twocolumn
\baselineskip 1.00pc

\title{ Cache-Aided Non-Orthogonal Multiple Access\\for 5G-Enabled Vehicular Networks}
 \author{Sanjeev~Gurugopinath,~\IEEEmembership{Member,~IEEE},~Paschalis~C.~Sofotasios,~\IEEEmembership{Senior~Member,~IEEE}, Yousof~Al-Hammadi,~\IEEEmembership{Member,~IEEE},~and~Sami~Muhaidat,~\IEEEmembership{Senior~Member, IEEE}
	
\thanks{S. Gurugopinath is with the Department of Electronics and Communication Engineering, PES University, Bengaluru 560085, India, (email: {\rm sanjeevg@pes.edu}).}

	\thanks{P. C. Sofotasios is with the Center for Cyber-Physical Systems, Department of Electrical and Computer Engineering, Khalifa University of Science and Technology, PO Box 127788, Abu Dhabi, UAE, and also with the Department of Electronics and Communications Engineering, Tampere University of Technology, 33101 Tampere, Finland (email: {\rm p.sofotasios@ieee.org}).}
	
		\thanks{Y. Al-Hammadi is  with the Department of Electrical and Computer Engineering, Khalifa University of Science and Technology, PO Box 127788, Abu Dhabi, UAE (email: {\rm yousof.alhammadi@ku.ac.ae}).}

	\thanks{S.  Muhaidat is with the  Center for Cyber-Physical Systems, Department of Electrical and Computer Engineering, Khalifa University of Science and Technology, PO Box 127788, Abu Dhabi, UAE and with the Institute for Communication Systems, University of Surrey, GU2 7XH, Guildford, UK, (email: {\rm muhaidat@ieee.org}).}
 }

\maketitle

\begin{abstract}
The increasing demand for rich multimedia services and the emergence of the  Internet-of-Things (IoT) pose challenging requirements  for the next generation vehicular networks. 
Such challenges are largely related to high spectral efficiency and low latency requirements in the context of massive content delivery and increased connectivity.
 In this respect, caching and non-orthogonal multiple access (NOMA) paradigms have been recently proposed as potential solutions to effectively address some of these key challenges. 
 In the present contribution, we introduce cache-aided NOMA as an enabling technology for vehicular networks.  
 In this context, we first  consider the full file caching case, where each vehicle caches and requests entire files using the NOMA principle.  
 Without loss of generality, we consider a two-user vehicular network communication scenario under  double Nakagami$-m$ fading conditions and propose an optimum power allocation policy. 
 To this end, an optimization problem that maximizes the overall probability of successful decoding of files at each vehicle is formulated and solved.
   Furthermore, we consider the case of split file caching, where each file is divided into two parts.
    A joint power allocation optimization problem is formulated, where  power allocation across vehicles and cached split files is investigated.
    The offered analytic results are corroborated by extensive results from computer simulations and interesting insights are developed. 
   Indicatively, it is shown that  the proposed caching-aided NOMA outperforms the conventional NOMA technique.
\end{abstract}

\begin{IEEEkeywords}
Caching, double Nakagami$-m$ fading channels, non-orthogonal multiple access, vehicle-to-vehicle communications, vehicular networks.
\end{IEEEkeywords}

\IEEEpeerreviewmaketitle

\section{Introduction} \label{SecIntro}

Recent research advances in information and wireless technologies have led to growing interest in the development of  intelligent transportation systems (ITS), which  promise significant improvements in road safety and traffic flow, and will enable new data services \cite{Tuohy_ITS_2015, Gerla_WF-IoT_2014}. In this context, vehicular networks consist of cooperative communication terminals that relay information with each other and also exchange information with fixed infrastructure, e.g., road side unit (RSU).  Potential applications of vehicular networks are diverse and pervasive; for example, transportation can be improved through fast dissemination of road and traffic information  and coordination of vehicles at critical points such as highway entries and other intersections. In addition, numerous challenging new applications of interest can be realized, e.g., high-speed internet access, cooperative downloading, network gaming among passengers of adjacent vehicles, and virtual, video-enabled meetings among co-workers in different vehicles.

As a dedicated short-range communications (DSRC) technology, wireless access for vehicular environments (WAVE), IEEE $802.11$p, enables data rates in the range  $6-27$ Mbps over short distances \cite{Vinel_ComLet_2012}.  Furthermore, it provides high-speed vehicle-to-vehicle (V2V) and vehicle-to-infrastructure (V2I) data transmission.
In this context, the 3rd generation partnership project (3GPP) recently  published V2X (vehicle-to-everything) specifications based on LTE as the underlying technology. However, in the also recently proposed internet-of-vehicles (IoV) ecosystem, the vehicles are envisioned to share and access enormous data quickly from the cloud, and are expected to be able to process them with high performance and  marginal overhead. 
Moreover, due to the high mobility scenarios, network topology can change quickly and the handover  among different base stations and vehicles becomes  more frequent. This poses several challenges to the effective  realization of reliable communications with low latency in both WAVE and LTE based vehicular networks. 
Therefore, various 5G enabling technologies, such as multiple-input multiple-output (MIMO), heterogeneous network (HN)-based communications, millimetre-wave (mmWave) communications and ultra-wideband (UWB) communications can be considered  to improve the overall efficiency of a vehicular network. 

A major challenge related to the above scheme  is that these technologies require  a significant amount of backhaul overhead.  
Yet,  it has been recently shown that caching techniques, in which popular files are stored a priori at  vehicles, are known to reduce the backhaul traffic. The basic idea of caching is to store the popular contents in different geographical locations during the off-peak time. This local storage arrangement enables a duplicate transmission of popular files, thereby increasing the spectral efficiency of the vehicular network and, subsequently, reducing latency. Also, it has been highlighted in the literature that installing memory units at the user-end is significantly  cost effective   compared to increasing the backhaul overhead \cite{Li_CST_2018}. 
Based on this, numerous investigations on  the integration of several caching techniques with technologies such as heterogeneous networks \cite{Maddah-Ali_TCOMMag_2016, Yang_TWC_2016, Cui_TWC_2017}, device-to-device communications \cite{Ji_TIT_2016}, \cite{Gregori_JSAC_2016}, and with cloud- and fog-radio access networks \cite{Tondon_ISIT_2016} have been carried out. 

It is recalled that non-orthogonal multiple access (NOMA) has been recently proposed as a promising multiple access solution to address some of the challenges in 5G networks (\cite{Liu_IEEE_2017, Wan_TWC_2018, Islam_TWC_2018, Kader_DSP_2019}).  In particular, NOMA is envisioned to increase the system throughput and to  support massive connectivity. 
More recently, it was shown that NOMA is capable of offering  performance benefits in visible light communication systems, where enhanced performance gains in terms of metrics such as spectral and energy efficiencies were reported \cite{Islam_5GTechFocus_2017}, \cite{Marshoud_TWC_2018}. Similarly, in the context of vehicular networks, it   allows different vehicles to share the same time and frequency resources \cite{Chen_JSAC_2017}. 
It is recalled that NOMA can be realized via two different approaches, namely, (a) power domain, and (b) code domain \cite{Ding_JSAC_2017}. In the power domain NOMA (PD-NOMA), users are assigned different power levels over the same resources; on the contrary,  in the code domain NOMA (CD-NOMA), multiplexing is carried out using spreading sequences for each user, similar to the code division multiple access (CDMA) technology. In this work, by NOMA, we refer to the PD-NOMA.

It is noted that NOMA can be realized by allocating  different power levels to different users, according to their channel conditions. This power allocation scheme enables the efficient implementation of successive interference cancellation (SIC), which is often considered to remove multi-user interference in multi-user detection systems. In this case, SIC is carried out  at users with the best channel conditions and is performed in descending order of the channel. 
It is also pointed out that significant performance improvements have been shown in vehicular networks due to the amalgamation of the NOMA technology with other techniques such as MIMO communications \cite{Chen_JSAC_2017}, \cite{Zheng_WCL_2017}, mmWave communications and spatial modulation \cite{Chen_JSAC_2017}. 
In what follows, we briefly describe some of the related techniques, and then highlight the contributions of the present analysis. 

\subsection{Related Work and Motivation} \label{SubSecRelatedWork}

\subsubsection{Caching} \label{SubSubSecCaching}

The advantages and applications of caching for cellular systems have been extensively investigated in the recent literatare \cite{Li_CST_2018}, \cite{Wang_CST_2018}. 
In this context, several metrics have been proposed and analyzed  to characterize the performance of caching algorithms, such as hitting rate \cite{Zaidi_ICCW_2015}, outage/coverage probability \cite{Bastug_ISWCS_2014}, area spectral efficiency \cite{Zhang_TCOM_2014}, and average throughput \cite{Ji_TIT_2016}. Likewise, cooperative caching in vehicular ad hoc networks (VANETs) was discussed  in \cite{Glass_CST_2017}, where it  was argued that cooperative caching in VANETs rely on three major components, namely discovery, management and consistency. 
The main task in the discovery component is to search the presence of a file requested by a vehicle in the network using, e.g., broadcasting techniques such as split caching \cite{Majd_GLOBECOM_2014} and shared caching \cite{Caetano_ICNC_2010}. The next task is carried out through the management component, which is to decide on (a) whether to cache the requested file, and (b) where to cache it, for a smoother duplicate transmission. Lastly, the task driven by the consistency component is to decide on how long the cached content should be retained in the network.

\subsubsection{Non-Orthogonal Multiple Access} \label{SubSubSecNOMA}

It is recalled that NOMA is considered as one of the enabling technolgies for 5G  \cite{Dai_TComMag_2015}, \cite{Han_Access_2018}. However,  the use of NOMA in vehicular networks has   received attention only  recently, where it has been shown that it significantly outperforms the traditional orthogonal multiple access based systems \cite{Di_GLOBECOM_2017}. 
Based on this, the spectral efficiency and resource allocation for NOMA based vehicle-to-everything (V2X) systems have been investigated in \cite{Di_TWC_2017}. 
Then, graph theory-based encoders and reliable decoders using the belief propagation algorithm for NOMA-based V2X communication systems was addressed in \cite{Khoueiry_VTC_2017},  along with the corresponding   outage capacity analysis.
 In \cite{Qian_JSAC_2017}, the performance of NOMA-enabled vehicle-to-small-cell (V2S), which constitutes a special case of V2I, has been investigated in the context of analyzing the joint optimization of cell association and power control for weighted sum rate maximization.
  The utility of NOMA along with spatial modulation technique was proposed for V2V MIMO communication in \cite{Chen_JSAC_2017}, and the achievable throughput as well as the corresponding optimal power allocation policies were derived.

\subsubsection{Caching-Aided NOMA} \label{SubSubSecCacheNOMA}

To the best of our knowledge, caching-aided NOMA systems have not been investigated in the context of vehicular networks. However, as mentioned earlier, there are some sporadic results reported on the advantages of cache-aided NOMA in cellular systems.
For example, it was shown in  \cite{Doan_ICC_2018} that exploited pre-cached contents can  reduce/mitigate the incurred interference,  while employing SIC for NOMA. An optimum power allocation for the considered network is investigated, aiming to maximize  the probability of successful decoding of files at each user. However, the simplistic  Rayleigh fading conditions were assumed, which is not practically realistic in  vehicular networks. Moreover, the model in \cite{Doan_ICC_2018} only considers the full file caching case; that is, the authors assume that in the caching phase, the files are cached as a whole, which is largely restrictive. On the contrary, only a split file caching framework was considered in \cite{Ding_TCOM_2018}. The optimum power allocation and the performance of the proposed system are characterized by the achievable rate region. Additionally, the framework in \cite{Ding_TCOM_2018} does not take into consideration the detrimental  impairments due to encountered fading effects.

\subsection{Contributions} \label{SubSecContributions}

We envision cache-aided  NOMA for vehicular networks  particularly useful because vehicles can exploit their own cache contents to reduce interference due to signals from other vehicles. That is, for a particular vehicle, if a file  is present in its own cache and requested by another vehicle, then it can utilize it  during the request phase in order to reduce the interference in the received  superimposed NOMA signal. Hence, the performance of a cache-aided NOMA system is improved by exploiting its  cache contents. Furthermore, the interference reduction due to caching also helps to reduce the computational complexity of  SIC. Although cache-aided NOMA for cellular networks has received attention in the recent literature \cite{Doan_ICC_2018}, \cite{Ding_TCOM_2018}, the available framework is not directly applicable to vehicular networks.  This is largely due the following challenges: high mobility, network partitioning, leading to several isolated clusters of nodes, rapidly changing network topology with intermittent connectivity, and the harsh propagation environments. The main focus of this work is the propagation environment, where the popular Rayleigh/Nakagami fading assumptions are no longer valid, due to mobility of vehicles and low elevations of their antennas. Consequently, in order to understand the full potential of cache-aided NOMA in vehicular networks, an in-depth analysis of the system performance under realistic channel models is required. To this end, contrary to the pioneering works in this area of research, we consider the cascaded (double) Nakagami fading channel model, which provides a realistic description of inter-vehicular channels. This can be justified by the fact that the received signal at a vehicle is due to a large number of signals reflected from statistically independent multiple scatterers \cite{Karagiannidis_TCOM_2007}, \cite{Ilhan_TVT_2009}. Accordingly, the double Nakagami fading model helps in studying the artifacts due to fading as a whole, in vehicular networks with high mobility.

In the present contribution, we introduce  cache-aided NOMA  for vehicular networks, over cascaded fading channels that follow the the double/cascaded Nakagami$-m$ distribution \cite{Karagiannidis_TCOM_2007}, \cite{Ilhan_TVT_2009}. 
Without loss of generality, we consider a two-user network, but  extension  to a multi-user network is straightforward. 
In this context, we consider the cases of (i) full file caching, where each vehicle can cache files as a whole during the caching phase \cite{Doan_ICC_2018};  and (ii) split file caching, where caching is performed in parts \cite{Ding_TCOM_2018}. 
For simplicity, we also consider the case of two split files, i.e, files are split into two parts, which can be readily extended  to the general case  of multiple part  file splitting. Under this setup, we characterize the performance of the proposed system in terms of the probability of successful decoding of intended files at both vehicles. In each case, we formulate an optimization problem to find the optimum power profile of each user  such that the overall probability of successful decoding is maximized. In the case of full file caching, we analytically show that the cost function is concave. 
The main contributions of this paper are summarized below:
\begin{itemize}
 \item We propose  cache-aided NOMA as a promising solution to address  the spectral efficiency requirements in 5G-enabled  vehicular networks.
 \item We analytically characterize the performance of the proposed cache-aided NOMA system under realistic double Nakagami$-m$ fading conditions in terms of probability of successful decoding of files at each vehicles.
 \item We consider the case of full file caching, where each vehicle can cache full set of requested files during the caching phase. We formulate an optimization problem to find the optimum power profile of each vehicle, such that the probability of successful decoding of targeted files at each vehicle is maximized. Under the effects of double Nakagami$-m$ fading conditions, we further show that the cost function of this optimization problem is concave.
 \item We  consider a practically relevant scenario of split-file caching, where a joint power allocation optimization problem is formulated  in order to find the optimum power allocation across vehicles and within each file. Through numerical analysis, we prove  that the cost function for this joint optimization problem is concave.
 \item We quantify the performance enhancement due to the proposed cache-aided NOMA over the conventional NOMA and cache-aided OMA. The impact of the double Nakagami$-m$ distributed fading effects    on the performance is also investigated.
\end{itemize}

To the best of the authors' knowledge, the above results have not been previously reported in the open literature. 

\subsection{Organization} \label{SubSecOrganization}

The remainder of this paper is organized as follows:  The system setup and the cache-aided NOMA architecture are described in Section~\ref{SecSysModel}. The optimal power allocation problem to maximize the probability of successful decoding of files in a two user vehicular network with full file caching is proposed and discussed in detail in Section~\ref{SecFullFile}. The principle underlying  the joint power allocation in the case where each requested file is split into two parts  and cached in order is presented in Section~\ref{SecSplitFile}. The corresponding numerical results and useful discussions are provided  in Section~\ref{SecRes}, while closing remarks  are given in Section~\ref{SecConc}.

\section{System Model} \label{SecSysModel}

\begin{figure}
\begin{center}
 \includegraphics[scale=0.5]{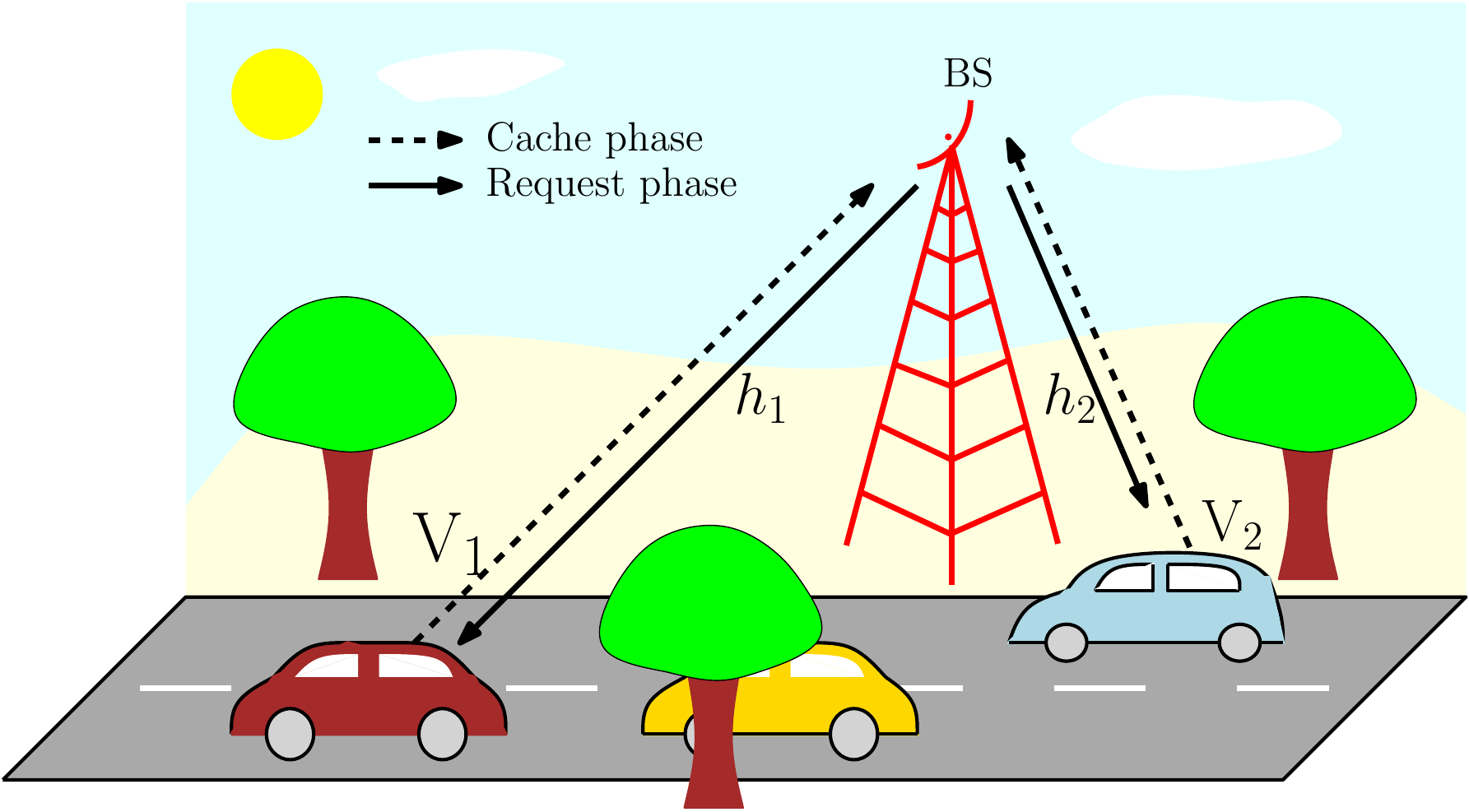}
 \caption{System model for the proposed caching-aided NOMA vehicular system, where a BS serves vehicles $V_1$ and $V_2$ simultaneously.} \label{FigNOMAv2vsysmodel}
\end{center}
\end{figure}

\subsection{Network Architecture} \label{SubSecNetArch}

We consider a vehicle-to-everything (V2X) network, which incorporates  V2I (vehicle-to-infrastructure), V2V (vehicle-to-vehicle), V2P (vehicle-to-pedestrian), and V2G (vehicle-to-grid) scenarios. As depicted in Fig.~\ref{FigNOMAv2vsysmodel} and without loss of generality,   we consider a base station (BS) serving two vehicles denoted by $V_1$ and $V_2$, respectively. Our analysis can also be extended to a network with more than two users, which can be realized by grouping them into pairs in clusters and  assigning orthogonal channels to each cluster. 
Also,  it is assumed that  $h_1$ and $h_2$ denote the complex channel coefficients between BS--$V_1$, and BS--$V_2$, respectively. It is recalled here that in vehicular communication systems\footnote{We consider a quasi-static fading channel between the BS and each users, typically observed in peak traffic hours \cite{Chen_JSAC_2017}.},  the channel gains $g_1 \triangleq |h_1|$ and $g_2 \triangleq |h_2|$ are   accurately modeled by a product of two independent cascaded Nakagami-$m$ distributions, i.e,   $g_1 \sim \mathcal{N}^2((\moneone, \omoneone), (\mtwoone, \omtwoone))$, and $g_2 \sim \mathcal{N}^2((\mtwoone, \omtwoone), (\mtwotwo, \omtwotwo))$, respectively, \cite{Karagiannidis_TCOM_2007}, \cite{Ilhan_TVT_2009}. 
To this effect, their probability density functions, denoted by $f_{g_1}(g)$ and $f_{g_2}(g)$, are given by \cite{Karagiannidis_TCOM_2007}
\begin{align}
f_{g_i}(g) = \frac{2 G_{0,2}^{2,0}\left(\left .\frac{\monei \mtwoi g^2}{\omonei \omtwoi} \right \rvert_{\monei, \mtwoi}^{\hspace{0.3cm} \text{------} } \right)}{g \Gamma(\monei) \Gamma(\mtwoi)},
\end{align}
for $i=~1,~2$, where $\Gamma(\cdot)$ and $G_{m,n}^{p,q}(\cdot)$ denote the  Euler's gamma and Meijer-G functions, respectively \cite{refGradBk}. 
The double Nakagami$-m$ distribution is considered a cascaded fading model and it constitutes a generalization of the double Rayleigh distribution, which is included as a special case \cite{Salo_TAP_2006}. Moreover,  the square of the double Nakagami$-m$ distributed random variables, that is, $g^2_1$ and $g^2_2$ have double Gamma distribution, and their respective densities are given   by a standard transformation of random variables  as:
\begin{align}
& f_{g^2_i}(g) = \frac{2 g^{\frac{\monei-\mtwoi}{2}-1}}{\Gamma(\monei) \Gamma(\mtwoi) \left(\frac{\omonei \omtwoi}{\monei \mtwoi}\right)^{\frac{\monei+\mtwoi}{2}}} \nonumber \\
& ~~~~~~~~~~~~~~~~~~~ \times \mathcal{K}_{\monei-\mtwoi} \left(2 \sqrt{\frac{\monei \mtwoi}{\omonei \omtwoi} g}\right), \label{CDFForV1}
\end{align}
for $i=1,~2$, where $\mathcal{K}_n(\cdot)$ is the modified Bessel function of the second kind, with order $n \in \mathbb{R}$ \cite{refGradBk}. 
Based on this,  the corresponding CDFs are given by
\begin{align}
F_{g^2_i}(g) = \frac{G_{1,3}^{2,1}\left(\left .\frac{\monei \mtwoi g}{\omonei \omtwoi} \right \rvert_{\monei, \mtwoi, 0}^{\hspace{0.5cm} \text{1} } \right)}{\Gamma(\monei) \Gamma(\mtwoi)},
\end{align}
for $i=1,~2$. 
In what follows, we introduce the model on caching.

\subsection{Caching Model and Assumptions} \label{SubSecCacheModel}
We assume  that $V_1$ and $V_2$ are equipped with a finite capacity cache of size $\kappa$. Let $\fbs \triangleq \{F_1, F_2, \ldots, F_T\}$ denote the finite set of files available at the BS. The popularity profile on $\fbs$ is then modeled by the popular Zipf distribution \cite{Li_CST_2018}, with a skewness control parameter $\zeta > 0$.\footnote{Our analysis is valid  even for any meaningful popularity profile. In general, the popularity profile is only used to dictate which files to cache at $V_1$ and $V_2$.} In particular, the popularity of a file $F_t \in \fbs$ is given by the probability
\begin{align}
q_t = \dfrac{1}{t^\zeta \sum_{i=1}^T \frac{1}{i^\zeta}}, ~~ t=1,\ldots,T,
\end{align}
such that $q_1 > q_2 > \cdots > q_T > 0$, and $\sum_{t=1}^T q_t = 1$. In the \emph{caching phase}, $V_1$ and $V_2$ fetch and store a set of files from $\mathrm{F}$, denoted by $\mathrm{F}_{V_1} \subset \mathrm{F}$ and $\mathrm{F}_{V_2} \subset \mathrm{F}$, during the off-peak hours. With the popularity profile considered above, the optimal caching policy would be to cache from the most popular file to the least popular file. In the \emph{requesting phase}, without loss of generality, we assume that users $V_1$ and $V_2$ request for files $F_1 \in \mathrm{F}$ and $F_2 \in \mathrm{F}$, respectively. Furthermore, it is  assumed that the BS has perfect knowledge of $\mathrm{F}_{V_1}$ and $\mathrm{F}_{V_2}$, and that $V_1$ and $V_2$ possess knowledge of the power allocation scheme at the BS. For the former assumption, the information can be obtained at the BS during the requesting phase, i.e., when $V_1$ and $V_2$  share their cache status. For the latter, the BS can broadcast the information to $V_1$ and $V_2$ before communicating their files.

\subsection{Downlink Communication Using NOMA} \label{SubSecNOMA}
Let $x_1$ and $x_2$ be the signals corresponding to $F_1$ and $F_2$. The BS uses the NOMA configuration  to send a superimposed signal to both $V_1$ and $V_2$, which can decode their respective signals using the successive interference cancellation (SIC) technique. Let $y_1$ and $y_2$ be the received signals at $V_1$ and $V_2$, respectively, which are given as:
\begin{align}
& y_1 = h_1 (\sqrt{\alpha P} x_1 + \sqrt{(1-\alpha)P} x_2) + w_1  
\end{align}
and
\begin{align}
 y_2 = h_2 (\sqrt{\alpha P} x_1 + \sqrt{(1-\alpha)P} x_2) + w_2,
\end{align}
where, $P$ is the average transmit power at the BS, $\alpha$ $\in (0,1)$ denotes the fraction of power assigned to $V_1$, and $w_1$ and $w_2$ are independent, circularly symmetric complex Gaussian distributions with zero mean and variance $\sigma_1^2$ and $\sigma_2^2$, respectively. Concisely, this is denoted as $w_1 \sim \mathcal{CN}(0,\sigma_1^2)$ and $w_2 \sim \mathcal{CN}(0,\sigma_2^2)$. As mentioned earlier, $|h_1| \sim \mathcal{N}^2((\moneone, \omoneone), (\mtwoone, \omtwoone))$ and $|h_2| \sim$ $\mathcal{N}^2((\monetwo, \omonetwo), (\mtwotwo, \omtwotwo))$.

It is important to observe that in the considered caching-based NOMA in vehicular networks, each user can decode its own signal by reducing or removing the interference due to the signal of the other user, by using the contents obtained during the caching phase. This increases the probability of successful decoding of the signals of each user. We assume that the files $F_1$ and $F_2$ can be decoded successfully by either $V_1$ or $V_2$, when the received SINR is greater than $\gamma_1 > 0$ and $\gamma_2 > 0$, respectively. Depending on the files cached by $V_1$ and $V_2$ during the caching phase, one of the following four scenarios can occur:
\begin{description}
 \item[Case~A:]~~$V_1$ has cached $F_2$, and $V_2$ has cached $F_1$.
 \item[Case~B:]~~$V_1$ has cached $F_2$, but $V_2$ has missed caching $F_1$.
 \item[Case~C:]~~$V_1$ has missed caching $F_2$, but $V_2$ has cached $F_1$.
 \item[Case~D:]~~$V_1$ has missed caching $F_2$, and $V_2$ has missed caching $F_1$.
\end{description}
Mathematically, the above  scenarios can be represented  as:
\begin{description}
 \item[Case~A:]~~$F_2 \in \mathrm{F}_{V_1}$, and $F_1 \in \mathrm{F}_{V_2}$.
 \item[Case~B:]~~$F_2 \in \mathrm{F}_{V_1}$, and $\{F_1\}\bigcap \mathrm{F}_{V_2} = \Phi$, where $\Phi$ is the null set.
 \item[Case~C:]~~$\{F_2\}\bigcap \mathrm{F}_{V_1} = \Phi$, and $F_1 \in \mathrm{F}_{V_2}$.
 \item[Case~D:]~~$\{F_2\}\bigcap \mathrm{F}_{V_1} = \Phi$, and $\{F_1\}\bigcap \mathrm{F}_{V_2} = \Phi$.
\end{description}
To this effect, we intend to find the optimal power allocation policy in each case, such that their respective probabilities of successful decoding of files $F_1$ and $F_2$ at $V_1$ and $V_2$ are maximized. In the next section, we provide the mathematical details of each case and formulate the corresponding optimization problems.

\section{Full File Caching: Optimal Power Allocation} \label{SecFullFile}
In this section, we provide the mathematical details on the optimal values of $\alpha$ for all the above mentioned cases, such that their respective probabilities of successful decoding of $F_1$ and $F_2$ at $V_1$ and $V_2$ are maximized. Towards this end, we first derive the expressions for the probability of successful decoding of files in each case.

\subsection{Probability of Successful Decoding} \label{SubSecProbSucc}

\subsubsection{Case A} \label{SubSecCaseA}
When both $V_1$ and $V_2$ have each others' files at their cache, irrespective of the value of $\alpha$, $V_1$ and $V_2$ will be able to decode $F_1$ and $F_2$ by canceling out the signals due to $F_2$ and $F_1$, respectively. Therefore, the probabilities of successful decoding of $F_1$ and $F_2$ at $V_1$ and $V_2$ are given by
\begin{align}
 \psvonea = \pr\left\{\frac{\alpha P g_1^2}{\sigma_1^2} > \gamma_1\right\}
\end{align}
and
\begin{align}
 \psvtwoa = \pr\left\{\frac{(1-\alpha) P g_2^2}{\sigma_2^2} > \gamma_2\right\} 
\end{align}
respectively. Recall that $\gamma_1$ and $\gamma_2$ denote the thresholds on SINR for successful decoding of files $F_1$ and $F_2$, at either $V_1$ or $V_2$.

\subsubsection{Case B} \label{SubSecCaseB}
As in the previous case, since $F_2 \in \mathrm{F}_{V_1}$, $V_1$ can decode $F_1$ irrespective of $\alpha$, with probability
\begin{align}
& \psvoneb = \pr\left\{\frac{\alpha P g_1^2}{\sigma_1^2} > \gamma_1\right\}.
\end{align}
However, since $V_2$ has not cached $F_1$, its probability of successful decoding depends on whether it is near or far from the BS. Each of these cases are considered separately. To this end, when $\alpha \leq 0.5$, that is, when $V_2$ is far from the BS, it treats the signal from $V_1$ as interference and decodes $F_2$ with a success probability  
\begin{align}
& \psvtwob = \pr\left\{\frac{(1-\alpha) P g_2^2}{\alpha P g_2^2 + \sigma_2^2} > \gamma_2\right\}.
\end{align}
On the contrary, when $\alpha > 0.5$, that is when $V_2$ is closer, it first decodes $F_1$ and then $F_2$, following SIC. The success probability in this case is given  by
\begin{align}
& \psvtwob = \pr\left\{\frac{\alpha P g_2^2}{(1-\alpha) P g_2^2 + \sigma_2^2} > \gamma_1\right\} \nonumber \\
& ~~~~~~~~~~~~~~~~ \times \pr\left\{\frac{(1-\alpha) P g_2^2}{\sigma_2^2} > \gamma_2\right\}.
\end{align}

\subsubsection{Case C} \label{SubSecCaseC}
Similar to the case B, irrespective of the value of $\alpha$, $V_2$ decodes $F_2$ successfully with a probability
\begin{align}
& \psvtwoc = \pr\left\{\frac{(1-\alpha) P g_2^2}{\sigma_2^2} > \gamma_2\right\}.
\end{align}
On the contrary, when $\alpha > 0.5$, $V_1$ decodes $F_1$ successfully with a probability
\begin{align}
& \psvonec = \pr\left\{\frac{\alpha P g_1^2}{(1-\alpha) P g_1^2+ \sigma_1^2} > \gamma_1\right\},
\end{align}
and with a corresponding success probability of
\begin{align}
& \psvonec = \pr\left\{\frac{(1-\alpha) P g_1^2}{\alpha P g_1^2 + \sigma_1^2} > \gamma_2\right\}  \pr\left\{\frac{\alpha P g_1^2}{\sigma_1^2} > \gamma_1\right\},
\end{align}
when $\alpha \leq 0.5$.

\subsubsection{Case D} \label{SubSecCaseD}
In this case, when $\alpha > 0.5$, $V_1$ is able to decode $F_1$ successfully with probability
\begin{align}
& \psvoned = \pr\left\{\frac{\alpha P g_1^2}{(1-\alpha) P g_1^2+ \sigma_1^2} > \gamma_1\right\},
\end{align}
while $V_2$ has a success probability of
\begin{align}
& \psvtwod = \pr\left\{\frac{\alpha P g_2^2}{(1-\alpha) P g_2^2 + \sigma_2^2} > \gamma_1\right\} \nonumber \\
& ~~~~~~~~~~~~~~~~ \times \pr\left\{\frac{(1-\alpha) P g_2^2}{\sigma_2^2} > \gamma_2\right\}.
\end{align}
Finally, when $\alpha \leq 0.5$, $V_2$ decodes $F_2$ with probability
\begin{align}
& \psvtwod = \pr\left\{\frac{(1-\alpha) P g_2^2}{\alpha P g_2^2 + \sigma_2^2} > \gamma_2\right\},
\end{align}
while the probability of success at $V_1$ is expressed  as
\begin{align}
& \psvoned = \pr\left\{\frac{(1-\alpha) P g_1^2}{\alpha P g_1^2 + \sigma_1^2} > \gamma_2\right\} \pr\left\{\frac{\alpha P g_1^2}{\sigma_1^2} > \gamma_1\right\}.
\end{align}

\subsection{Power Allocation} \label{SubSecPowerAlloc}

The optimal power allocation problems in all the considered cases is described below, recalling  that in each case, the probability of successful decoding of $F_1$ and $F_2$ is the product of individual success probabilities.

\subsubsection{Case A} \label{SubSecPACaseA}
As seen earlier for case A, the probability of successful decoding is independent of $\alpha$. Therefore, the optimization problem for case A would be
\begin{align}
&\opa: \max~ \psvonea \psvtwoa \nonumber \\
&\phantom{opa: }~~~ \text{s.t.}~~0 \leq \alpha \leq 1. \label{OptProbCaseA}
\end{align}

Next, we prove that the cost function in \eqref{OptProbCaseA} is concave in $0 \leq \alpha \leq 1$, through the following proposition. The concavity of the probability of success and the details corresponding to the other cases lead to lengthy expressions, and are omitted for brevity.

\begin{prop} \label{ThmMain}
The function $\psa \triangleq \psvonea \psvtwoa$ is concave in $0 \leq \alpha \leq 1$.
\end{prop}
\begin{proof}
See Appendix.
\end{proof}
A similar approach can be employed to establish the concavity of the cost function in the other optimization problems.

\subsubsection{Case B} \label{SubSecPACaseB}
For case B, recall that two sub-cases occur. For the case with $\alpha > 0.5$,  the requirement from $\psvtwob$, namely,
\begin{align}
\frac{\alpha P g_2^2}{(1-\alpha) P g_2^2 + \sigma_2^2} > \gamma_1  
\end{align}
i.e.
\begin{align}
\frac{\sigma_1^2}{\alpha P \left[\frac{1}{\gamma_1}-1\right] - P} > 0
\end{align}
yields the condition $\alpha > \frac{\gamma_1}{1+\gamma_1}$, and $\psvtwob$ also results in the requirement that $\alpha \leq 1$. Since $\gamma_2 > 0$, we obtain the following optimization problem: 
\begin{align}
&\opbone: \max~~~~ \psvoneb \psvtwob \nonumber \\
&\phantom{\opbone: }~~~\text{s.t.}~~\frac{\gamma_1}{1+\gamma_1} \leq \alpha \leq 1,
\end{align}
Similarly, when $\alpha \leq 0.5$, the condition
\begin{align}
\frac{(1-\alpha) P g_2^2}{\alpha P g_2^2 + \sigma_2^2} > \gamma_2 
\end{align}
i.e.
\begin{align}
\frac{\sigma_2^2}{\frac{P}{\gamma_2^2}-\alpha P\left(\frac{1}{\gamma_2}+1\right)} > 0
\end{align}
dictates that $\alpha \leq \frac{1}{1+\gamma_2}$. In this case, the optimization problem becomes
\begin{align}
&\opbtwo: \max~~~ \psvoneb \psvtwob \nonumber \\
&\phantom{\opbtwo: }~~~\text{s.t.}~~0 \leq \alpha \leq \frac{1}{1+\gamma_2}.
\end{align}
with the requirement $\alpha > 0$ from $\psvoneb$.

\subsubsection{Case C} \label{SubSecPACaseC}
The individual optimization problems corresponding to $\alpha > 0.5$ and $\alpha \leq 0.5$ are respectively given by
\begin{align}
&\opcone: \max~~~ \psvonec \psvtwoc \nonumber \\
&\phantom{\opcone: }~~~\text{s.t.}~~\frac{\gamma_1}{1+\gamma_1} \leq \alpha \leq 1, \\
&\opctwo: \max~~~~ \psvonec \psvtwoc \nonumber \\
&\phantom{\opctwo: }~~~\text{s.t.}~~0 \leq \alpha \leq \frac{1}{1+\gamma_2}.
\end{align}

\subsubsection{Case D} \label{SubSecPACaseD}
Finally, the optimization problem for case D for $\alpha > 0.5$ and $\alpha \leq 0.5$ can be formulated  as
\begin{align}
&\opdone: \max~~~ \psvoned \psvtwod \nonumber \\
&\phantom{\opdone: }~~~\text{s.t.}~~\frac{\gamma_1}{1+\gamma_1} \leq \alpha \leq 1, \\
&\opdtwo: \max~~~~ \psvoned \psvtwod \nonumber \\
&\phantom{\opdtwo: }~~~\text{s.t.}~~0 \leq \alpha \leq \frac{1}{1+\gamma_2}.
\end{align}

Next, we extend the above framework to the case where the individual files $F_1$ and $F_2$ are split into two parts, and each part is transmitted simultaneously.

\section{Split File Caching: Optimal Power Allocation} \label{SecSplitFile}
In this section, we consider the case where the files $F_1$ and $F_2$ are split further into sub-files. 
Without loss of generality, we consider the simplest case where each file is split into two sub-files. 
This  analysis can then be readily  extended to the case of splitting a file to greater than two sub-files. As shown in Fig.~\ref{FigNOMAv2vPCS}, the sub-files of $F_1$ and $F_2$ are denoted by $\foneone$, $\fonetwo$, and $\ftwoone$, $\ftwotwo$, respectively. It is assumed that the BS needs $\beta$ fraction of the assigned powers to transmit $F_{\ell}^{(1)}$, $\ell=1, 2$ and $1-\beta$ fraction of the powers to transmit $F_{\ell}^{(2)}$, $\ell=1, 2$, respectively. Moreover, let the signals corresponding to $\foneone$, $\fonetwo$, $\ftwoone$, and $\ftwotwo$ be denoted by $\xoneone$, $\xonetwo$, $\xtwoone$, and $\xtwotwo$, respectively. 
Based on this, the superimposed signal transmitted by the BS is represented as 
\begin{align}
& x = \sqrt{\beta \alpha P} \xoneone + \sqrt{(1 \hspace{-0.1cm} - \hspace{-0.1cm} \beta) \alpha P} \xonetwo \nonumber \\
& ~~~~~~~~~~ + \sqrt{\beta (1 \hspace{-0.1cm} - \hspace{-0.1cm} \alpha) P} \xtwoone + \sqrt{(1 \hspace{-0.1cm} - \hspace{-0.1cm} \beta) (1 \hspace{-0.1cm} - \hspace{-0.1cm} \alpha) P} \xtwotwo.
\end{align}
Let the SINR constraint for decoding the files $\foneone$, $\fonetwo$, $\ftwoone$, and $\ftwotwo$ be given by $\goneone$, $\gonetwo$, $\gtwoone$, and $\gtwotwo$, respectively. Depending on the file portions cached by $V_1$ and $V_2$, several classes and sub-cases are conceivable. An example of one such class includes the following sub-cases:
\begin{enumerate}[(a)]
 \item $V_1$ has cached $\ftwoone$ and $V_2$ has cached $\foneone$.
 \item $V_1$ has not cached any portions of $F_2$ and $V_2$ has cached $\foneone$.
 \item $V_1$ has cached $\ftwoone$ and $V_2$ has not cached any portions of $F_1$.
 \item Neither $V_1$ has cached any portions of $F_2$, nor $V_2$ has cached any portions of $F_1$.
\end{enumerate}

In the following, we restrict our attention to the first scenario, which is illustrated in Fig.~\ref{FigNOMAv2vPCS}, among such several other cases which can be defined accordingly \cite{Ding_TCOM_2018}. Our framework can also be similarly extended to the other cases. For the case at hand, the signals received at $V_1$ and $V_2$, after cancellation of pre-cached signals are given by
\begin{align}
& y_1 = h_1 (\sqrt{\beta \alpha P} \xoneone + \sqrt{(1 \hspace{-0.1cm} - \hspace{-0.1cm} \beta) \alpha P} \xonetwo \nonumber \\
& ~~~~~~~~~~ + \sqrt{(1 \hspace{-0.1cm} - \hspace{-0.1cm} \beta) (1 \hspace{-0.1cm} - \hspace{-0.1cm} \alpha) P} \xtwotwo) + w_1 
\end{align}
and
\begin{align}
& y_2 = h_2 (\sqrt{(1 \hspace{-0.1cm} - \hspace{-0.1cm} \beta) \alpha P} \xonetwo + \sqrt{\beta (1 \hspace{-0.1cm} - \hspace{-0.1cm} \alpha) P} \xtwoone \nonumber \\
& ~~~~~~~~~~ + \sqrt{(1 \hspace{-0.1cm} - \hspace{-0.1cm} \beta) (1 \hspace{-0.1cm} - \hspace{-0.1cm} \alpha) P} \xtwotwo) + w_2. 
\end{align}

\begin{figure}[ht]
\begin{center}
 \includegraphics[scale=0.7]{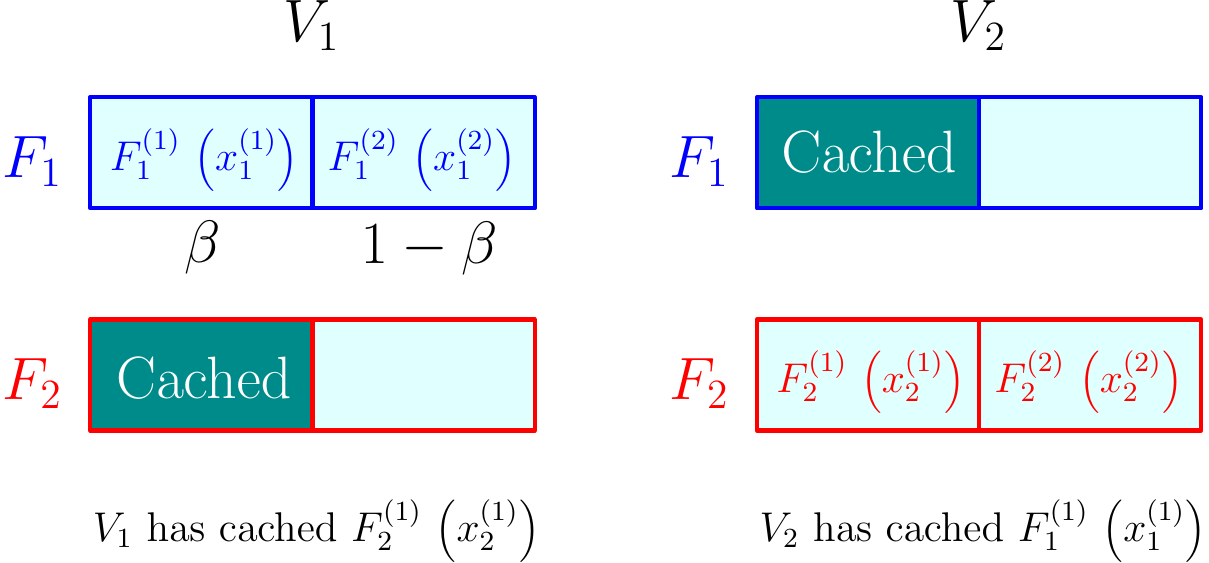}
 \caption{File split-up model for transmission and caching. The right hand side of the figure shows the case where $\ftwoone \in \mathrm{F}_{V_1}$ and $\foneone \in \mathrm{F}_{V_2}$.} \label{FigNOMAv2vPCS}
\end{center}
\end{figure}

For the above case where $V_1$ has cached $\ftwoone$ and $V_2$ has cached $\foneone$, that is, $\ftwoone \in \mathrm{F}_{V_1}$, and $\foneone \in \mathrm{F}_{V_2}$, two sub-cases arise depending on whether $\alpha \leq 0.5$ or $\alpha > 0.5$. For the case when $\alpha > 0.5$, the far user $V_1$ will  first decode $\foneone$, with probability
\begin{align}
\psvoneaone = \pr \left\{ \frac{\alpha \beta P g_1^2}{(1-\beta) P g_1^2 + \sigma_1^2} > \goneone  \right\},
\end{align}
and will decode $\fonetwo$ later, with probability
\begin{align}
\psvoneatwo = \pr \left\{ \frac{\alpha (1-\beta) P g_1^2}{(1-\alpha) (1-\beta) P g_1^2 + \sigma_1^2} > \gonetwo  \right\}.
\end{align}
On the contrary, $V_2$, which is the near user, employs SIC and first decodes $\fonetwo$ with probability
\begin{align}
\psvtwoaone = \pr \left\{ \frac{\alpha (1-\beta) P g_2^2}{(1-\alpha) P g_2^2 + \sigma_2^2} > \gonetwo\right\},
\end{align}
then decodes $\ftwoone$ with probability
\begin{align}
\psvtwoatwo = \pr \left\{ \frac{\beta (1-\alpha) P g_2^2}{(1-\alpha)(1-\beta) P g_2^2 + \sigma_2^2} > \gtwoone \right\},
\end{align}
followed by decoding $\ftwotwo$ with probability
\begin{align}
\psvtwoathree = \pr \left\{ \frac{(1-\beta) (1-\alpha) P g_2^2}{\sigma_2^2} > \gtwotwo \right\}.
\end{align}

Similarly, for the case when $\alpha \leq 0.5$, the far user $V_2$ will decode $\ftwoone$ first, with probability
\begin{align}
\psvtwoaone = \pr \left\{ \frac{(1-\alpha) \beta P g_2^2}{(1-\beta) P g_2^2 + \sigma_2^2} > \gtwoone  \right\},
\end{align}
followed by $\ftwotwo$ with probability
\begin{align}
\psvtwoatwo = \pr \left\{ \frac{(1-\alpha) (1-\beta) P g_2^2}{\alpha (1-\beta) P g_2^2 + \sigma_2^2} > \gtwotwo  \right\}.
\end{align}
Lastly, the near user $V_1$ employs SIC to decode $\ftwotwo$ first with probability
\begin{align}
\psvoneaone = \pr \left\{ \frac{(1-\alpha) (1-\beta) P g_1^2}{\alpha P g_1^2 + \sigma_1^2} > \gtwotwo  \right\},
\end{align}
then decodes $\foneone$ with probability
\begin{align}
\psvoneatwo = \pr \left\{ \frac{\alpha \beta P g_1^2}{\alpha (1-\beta) P g_1^2 + \sigma_1^2} > \goneone  \right\},
\end{align}
followed by finally decoding $\fonetwo$ with probability
\begin{align}
\psvoneathree = \pr \left\{ \frac{\alpha (1-\beta) P g_1^2}{\sigma_1^2} > \gonetwo  \right\}.
\end{align}

\subsection{Power Allocation} \label{SubSecPASplit}

Following the above discussion and the framework considered in Sec.~\ref{SecFullFile}, we propose the following optimization problem to jointly design the power allocation factors $\alpha$ and $\beta$, such that the overall probability of successful decoding is maximized. That is, for the case when $\alpha > 0.5$, the optimization problem is formulated as
\begin{align}
&\opaa: \max_{\alpha, \beta}~ \psvoneaone \psvoneatwo \psvtwoaone \psvtwoatwo \psvtwoathree \nonumber \\
&\phantom{opaa: }~ \text{s.t.}~~0.5 \leq \alpha \leq 1, \nonumber \\
&\phantom{opaa: ~ \text{s.t.}}~~0 \leq \beta \leq 1 \label{EqnSplitFile_alphagtrp5}
\end{align}
whereas for the case of $\alpha \leq 0.5$, the optimization problem is formulated as
\begin{align}
&\opaa: \max_{\alpha, \beta}~ \psvoneaone \psvoneatwo \psvoneathree \psvtwoaone \psvtwoatwo \nonumber \\
&\phantom{opaa: }~ \text{s.t.}~~0 \leq \alpha \leq 0.5, \nonumber \\
&\phantom{opaa: ~ \text{s.t.}}~~0 \leq \beta \leq 1. \label{EqnSplitFile_alphalsrp5}
\end{align}

It is not an easy task  to analytically characterize the optimization problems given above, due to the complicated expressions in \eqref{CDFForV1}. Therefore, we resort to numerical techniques, where we observe that the probability of successful decoding is concave with respect to the tuple $(\alpha, \beta)$ in their respective range. The optimal pair $(\alpha^*, \beta^*)$ can be obtained through techniques such as steepest ascent, and other search algorithms. This point is further elaborated in Sec.~\ref{SecRes}.

\begin{table}
\centering
\caption{Parameter Settings}
\label{table}
\setlength{\tabcolsep}{10pt}
\begin{tabular}
	{|p{110pt}|p{50pt}|}
\hline
			\textbf{Parameters}                                                         & \textbf{Settings }                 \\ \hline
		
			$\moneone, \mtwoone$                                       & $1, 1$                          \\
			$\monetwo, \mtwotwo$                                           & $1, 1$                          \\
			$\omoneone, \omtwoone$                        & $2, 2$                         
\\
			$\omonetwo, \omtwotwo$                     & $2, 2$                         
\\
			Path loss exponent $d$                                         & $2$                          \\
			Distances of $V_1$ and $V_2$ from BS                                         & $1, 0.5$                          \\
			Files at BS, $T$                                         & $5$                          \\
			Cache size at $V_1$ and $V_2$                                & $1$
\\
			Noise variance $\sigma_1^2$                                        & $1$                         \\
			Noise variance $\sigma_2^2$ & $1$
\\
			SINR threshold $\gamma_1$  & $1$
\\
			SINR threshold $\gamma_2$                                     & $1$                     \\
			Zipf parameter $\zeta$                                & $0.5$                        \\ \hline
\end{tabular}
\label{tab1}
\end{table}

\section{Numerical Results} \label{SecRes}

In this section, we discuss the performance of the proposed cache-aided NOMA, and compare its performance to conventional NOMA  and cache-aided  orthogonal multiple access (OMA) systems. First, we consider the full file caching scenario studied in Sec.~\ref{SecFullFile}. Unless stated otherwise, the parameters used for our simulations are listed in Tab.~\ref{tab1}. Since the popularity profile is modeled by the Zipf distribution, the optimal caching policy would be to cache the most popular files, depending on the cache size of $V_1$ and $V_2$.

\begin{figure}[ht]
\begin{center}
 \includegraphics[width=9.2cm, height=6.5cm]{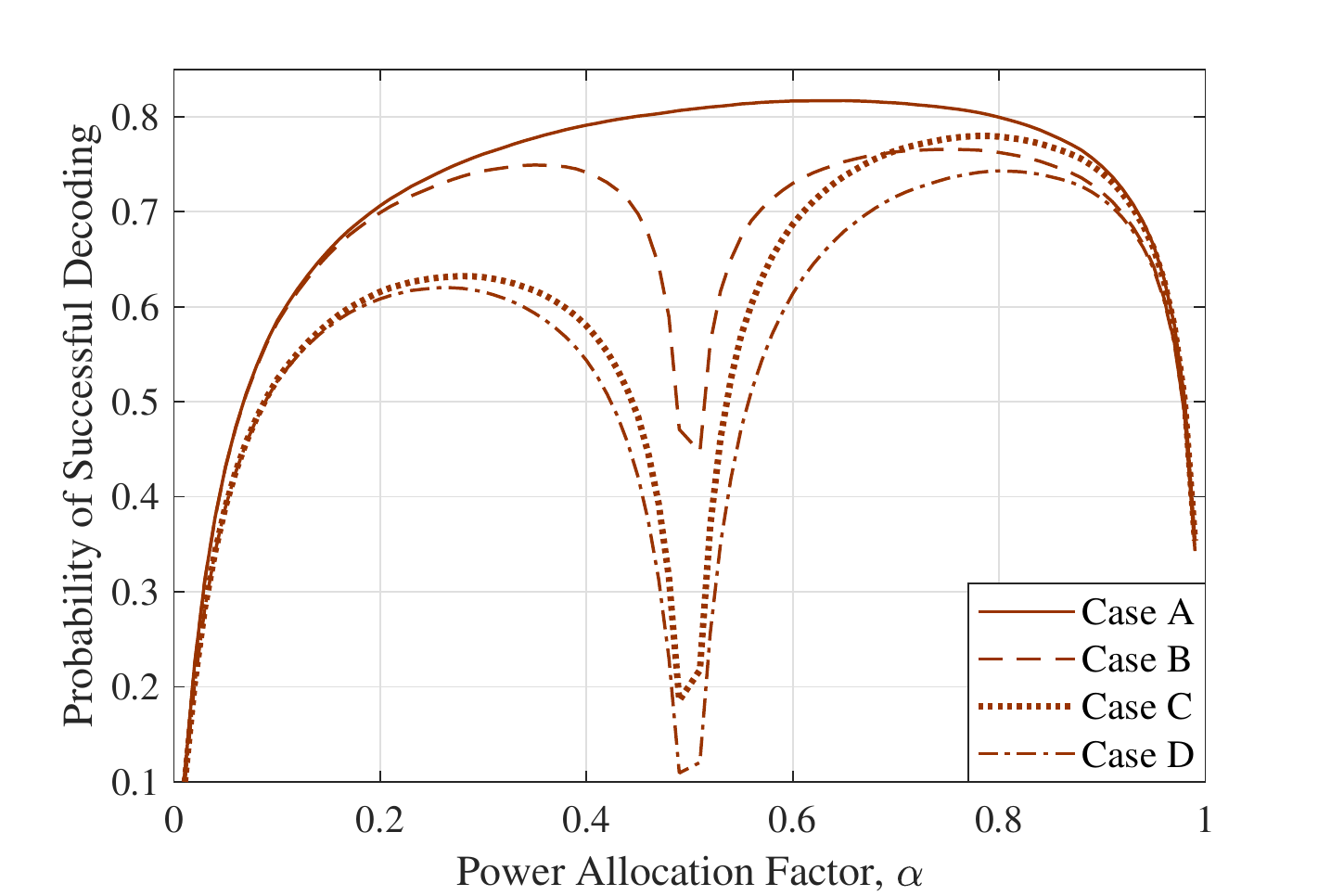}
 \caption{The cost functions in the associated optimization problems are shown to be concave.} \label{FigConcave}
\end{center}
\end{figure}

\begin{figure}[ht]
\begin{center}
 \includegraphics[width=9.2cm, height=6.5cm]{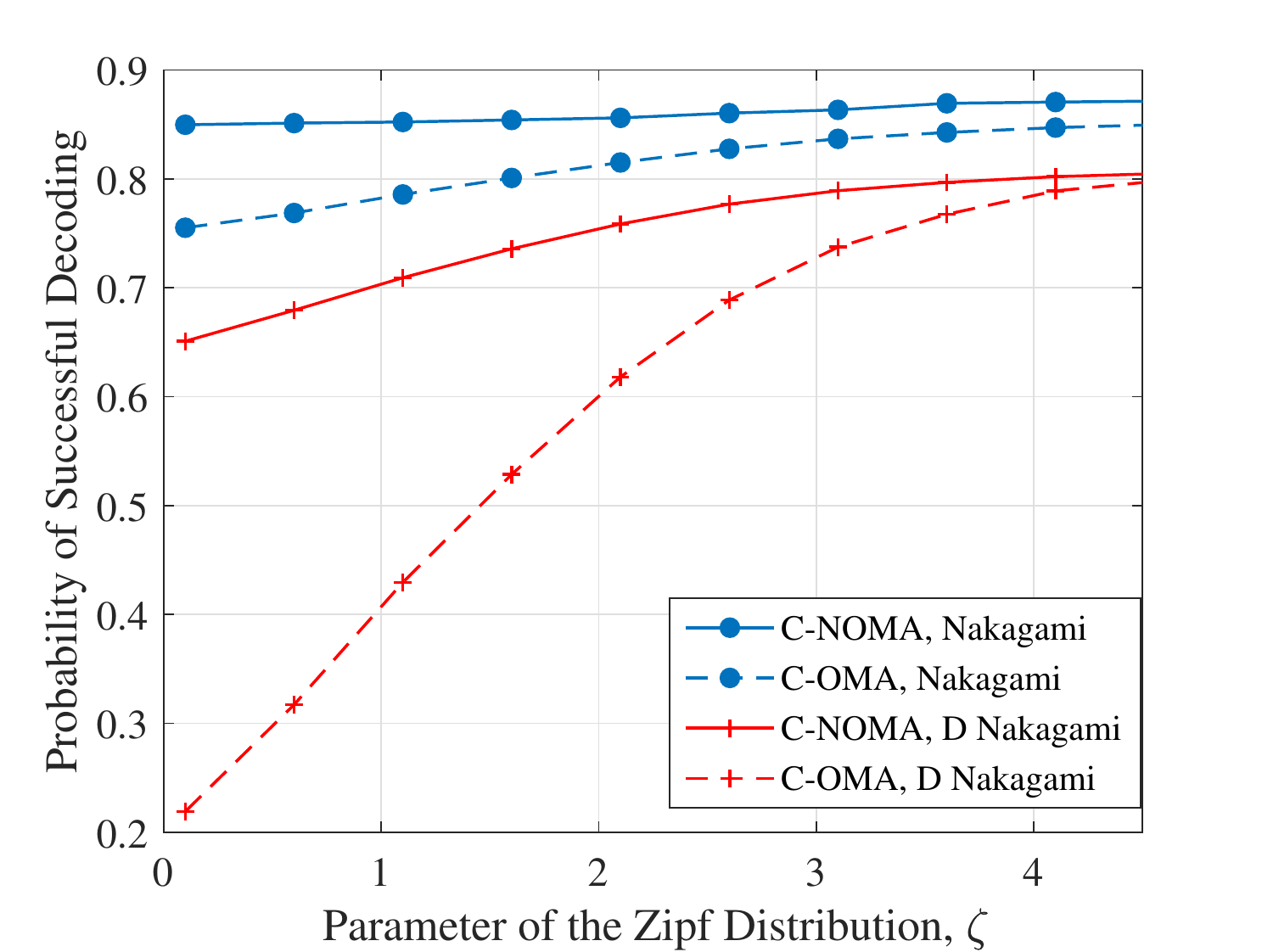}
 \caption{Performance comparison of caching-aided NOMA (denoted by C-NOMA) and caching-aided OMA (denoted by C-OMA), in terms of variation of probability of successful decoding with different values of the Zipf distribution parameter $\zeta$, under Nakagami and double Nakagami distributions.} \label{FigNvsNN_Zeta}
\end{center}
\end{figure}

\begin{figure}[ht]
\begin{center}
 \includegraphics[width=9.2cm, height=6.5cm]{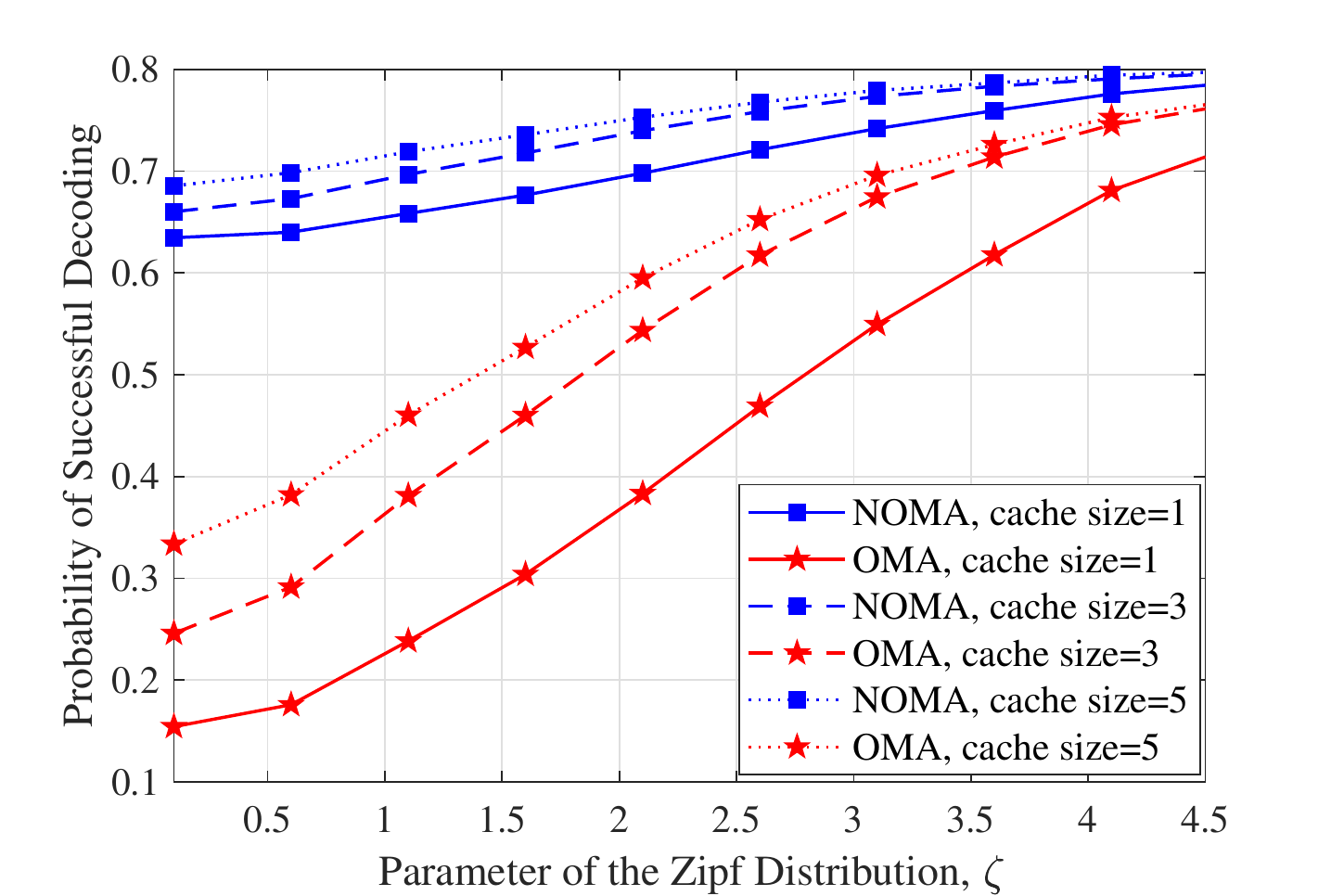}
 \caption{Performance comparison of caching-aided NOMA and caching-aided OMA, in terms of variation of probability of successful decoding with $\zeta$, for different values of cache sizes available at $V_1$ and $V_2$.} \label{FigPsucVszeta}
\end{center}
\end{figure}

\begin{figure}[ht]
\begin{center}
 \includegraphics[width=9.2cm, height=6.5cm]{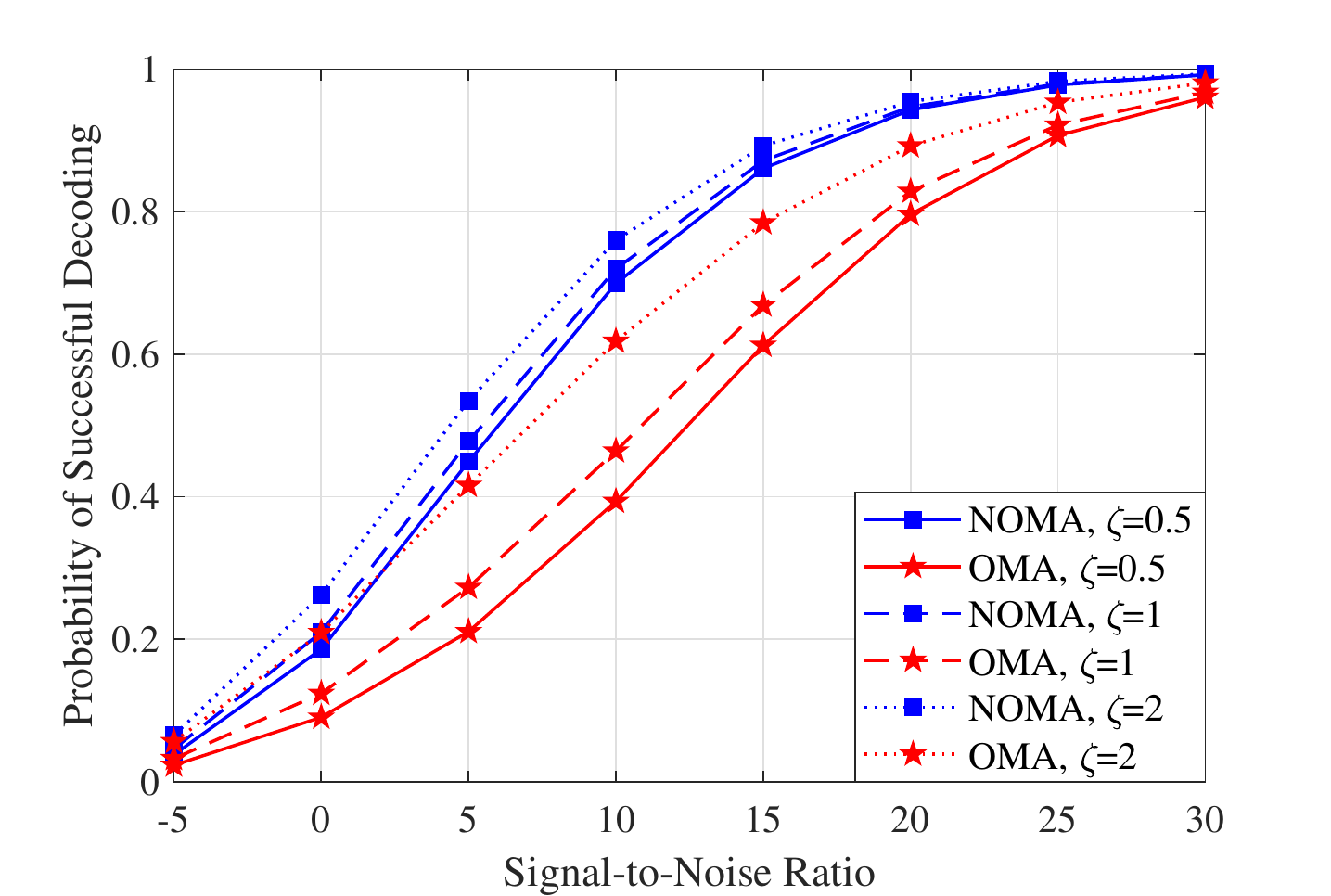}
 \caption{Performance comparison of caching-aided NOMA and caching-aided OMA, in terms of variation of probability of successful decoding with SNR, for different values of the Zipf parameter $\zeta$.} \label{FigPsucVsSNR}
\end{center}
\end{figure}

Fig.~\ref{FigConcave} shows the performance of the probability of successful decoding, i.e., $\psvonea \psvtwoa$ with respect to $\alpha$, for the cases A, B, C and D. It is noticed that for each case, as expected, the cost function is concave in $\alpha$. In particular for case A, as expected from Lemma~\ref{ThmMain}, the probability of successful decoding is concave in the entire range of $\alpha$. The optimal $\alpha$ in each case and sub-cases is found through the bisection method \cite[Alg.~4.1]{Boyd_CVX}.

Fig.~\ref{FigNvsNN_Zeta} compares the performances of cache-aided NOMA and cache-aided OMA systems under Nakagami$-m$ and double Nakagami$-m$ fading conditions, which are considered suitable in vehicle-to-vehicle communications. To this effect, different values of $\zeta$ are considered, with an SNR of $10$ dB. In both cases, NOMA exhibits a better performance compared to the OMA counterpart. 
It is evident though that for both NOMA and OMA, the performance under the double Nakagami$-m$ fading conditions is inferior  compared to that encountered under of conventional Nakagami$-m$ fading, which is expected due to the cascaded nature of the double Nakagami$-m$ model.  However, the magnitude of degradation in OMA is severe compared to NOMA, especially in the regime where $\zeta \rightarrow 0$, which is of practical relevance. Therefore, the assumption of Nakagami$-m$  fading model yields an unrealistic upper bound on the performances of cache-aided NOMA and cache-aided OMA, which is over-optimistic in the latter case.

Fig.~\ref{FigPsucVszeta} compares the performances of cache-aided NOMA and cache-aided OMA systems for different values of $\zeta$ and cache sizes, available at $V_1$ and $V_2$. The received SNR is fixed at $10$ dB and it  is observed that  a large cache at $V_1$ and $V_2$ improves the performance of NOMA and OMA, since the probability of caching the files requested by another vehicle increases with the cache size. Once again, for lower values of $\zeta$, it is observed that the NOMA system offers significant performance gain compared to the OMA counterpart, while it also exhibits a slower performance degradation with $\zeta$.

\begin{figure}[ht]
\begin{center}
 \includegraphics[width=9.2cm, height=6.5cm]{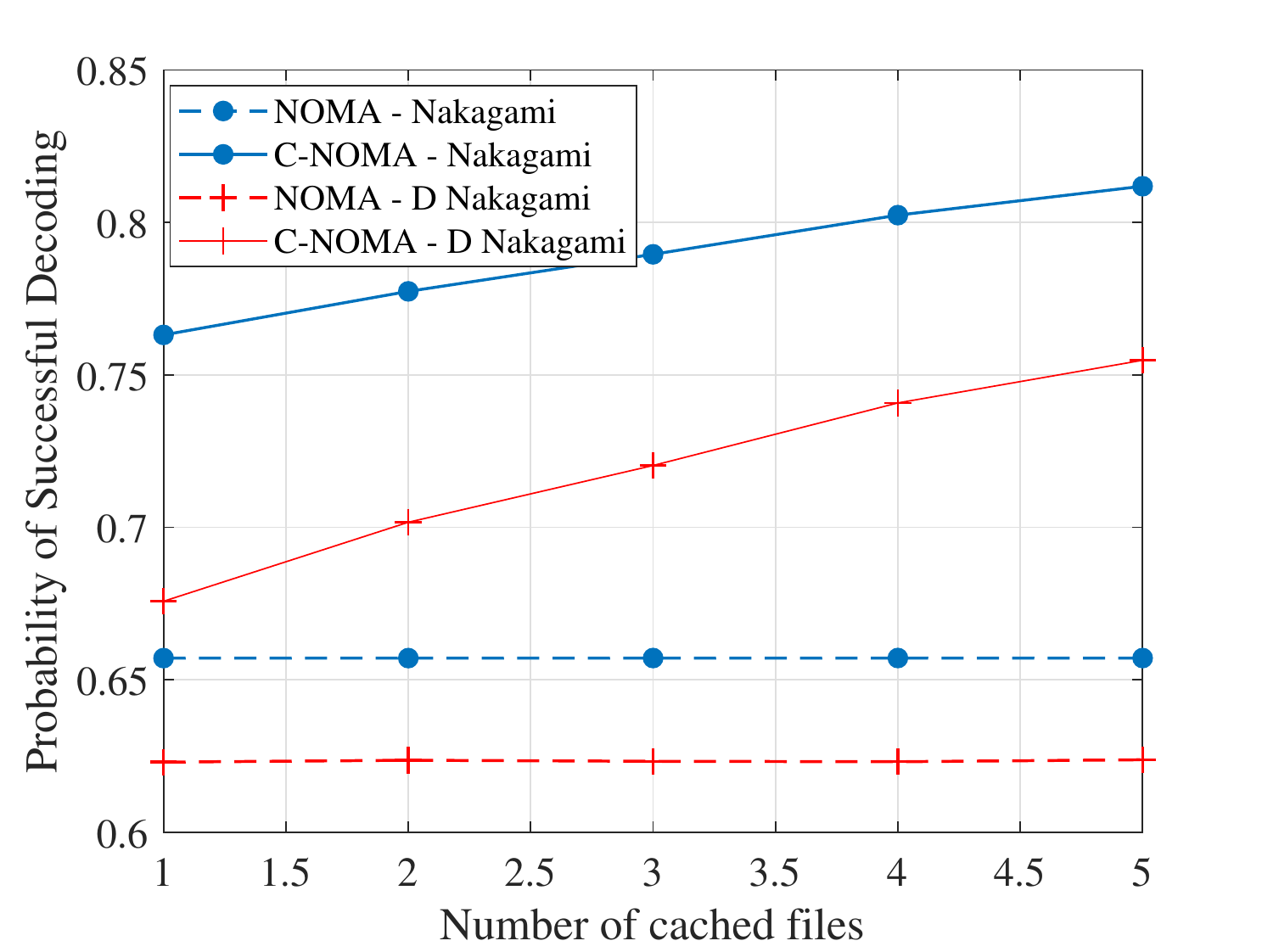}
 \caption{Performance comparison of caching-aided NOMA, in terms of variation of probability of successful decoding for different number of cached files, under Nakagami$-m$ and double Nakagami$-m$ distributions.} \label{FigNvsNN_CachedFiles}
\end{center}
\end{figure}

\begin{figure}[ht]
\begin{center}
 \includegraphics[width=9.2cm, height=6.5cm]{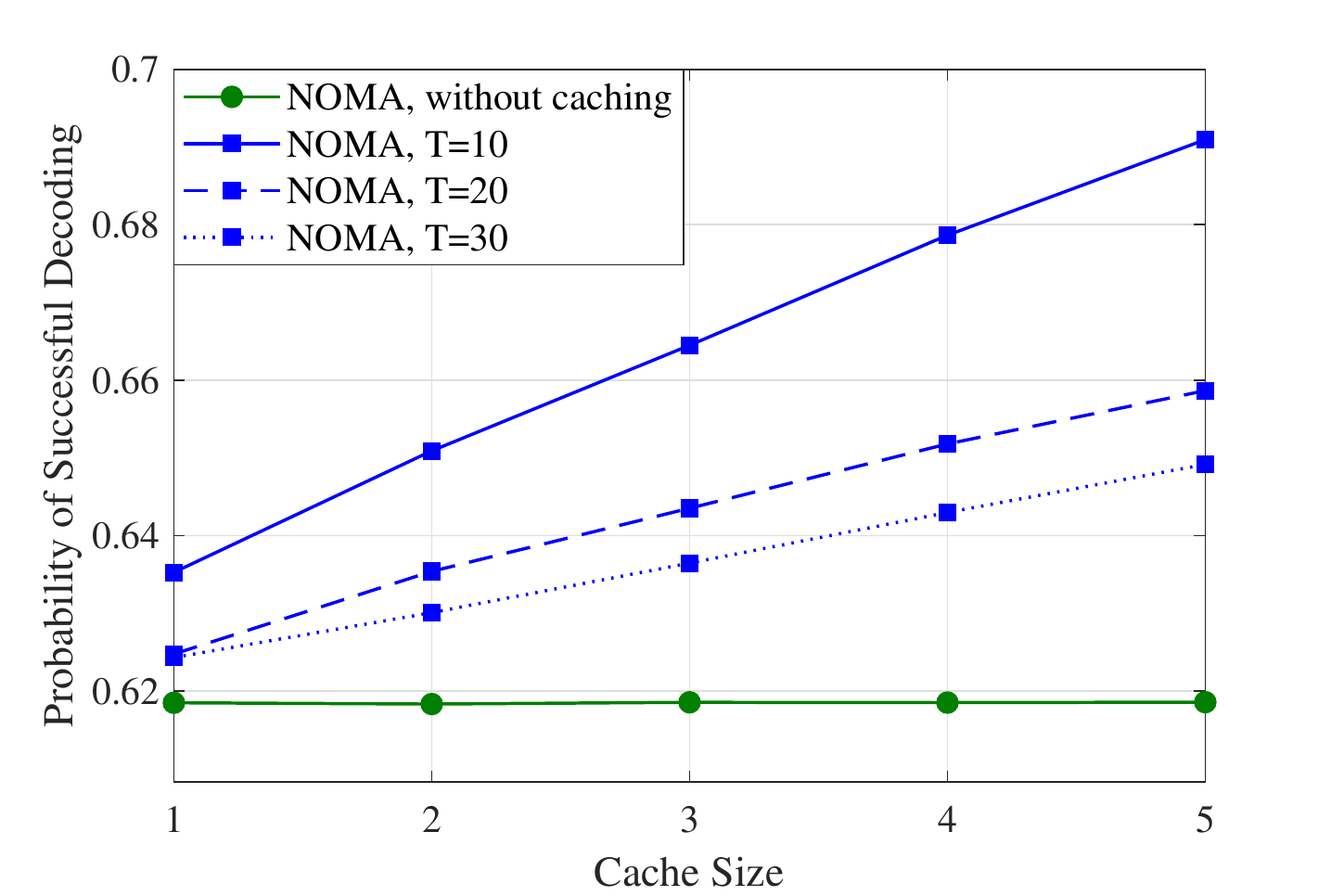}
 \caption{Performance comparison of caching-aided NOMA and conventional NOMA, in terms of variation of probability of successful decoding with different cache sizes at each vehicle, for different number of files available at the BS.} \label{FigPsucVsMaxCache}
\end{center}
\end{figure}

\begin{figure}[ht]
\begin{center}
 \includegraphics[width=9.2cm, height=6.5cm]{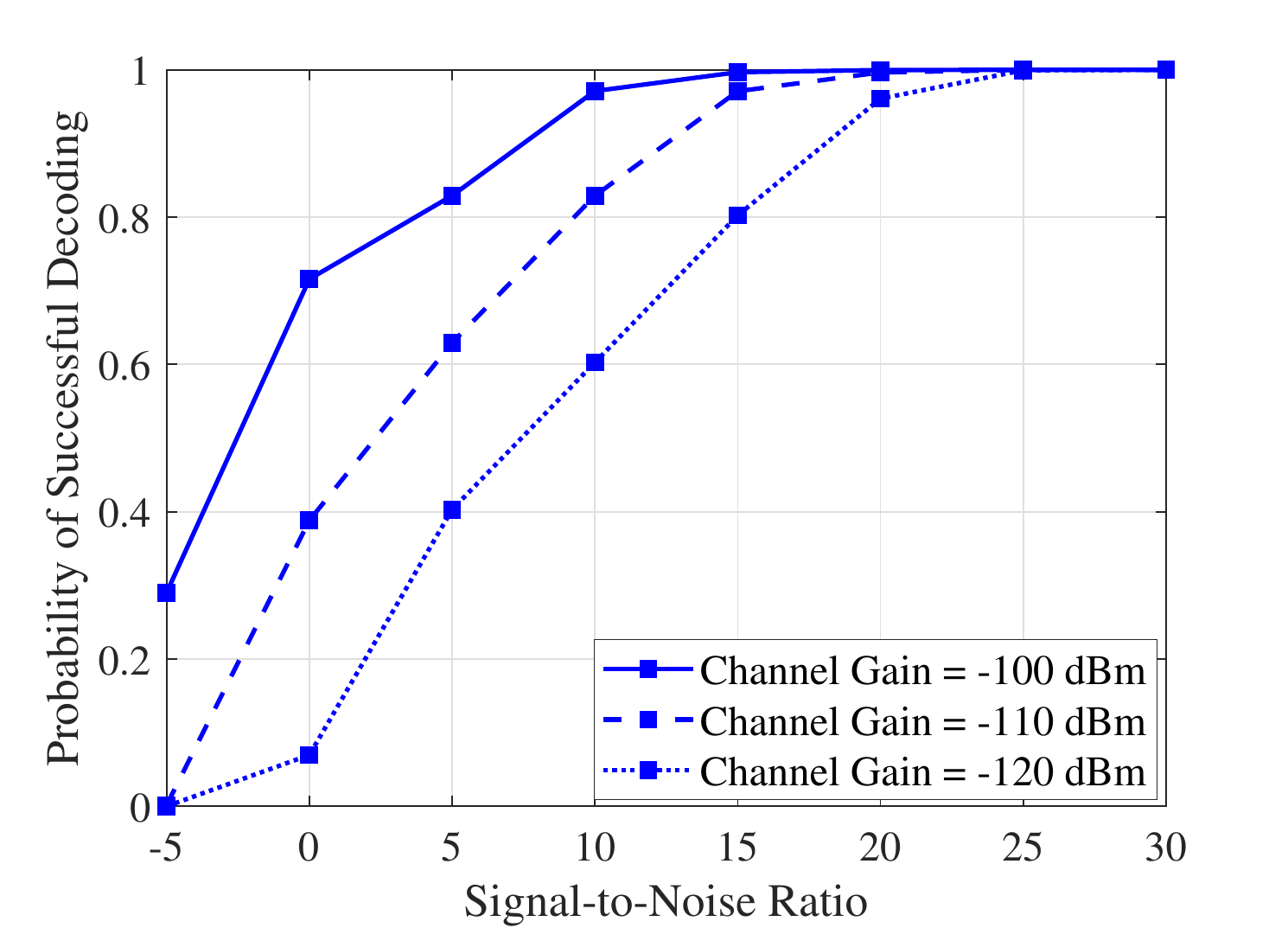}
 \caption{Performance comparison of caching-aided NOMA under spatio-temporally correlated Rician distribution \cite{Ozdogan_TWC_2019}, in terms of variation of probability of successful decoding with SNR, for different values of the channel gain at the antenna of each vehicle.} \label{FigPsucVsSNR_STCRice}
\end{center}
\end{figure}

Fig.~\ref{FigPsucVsSNR} depicts the performance of proposed scheme and compare it with  cache-aided OMA for different Zipf parameter $\zeta$ values. As expected, the proposed cache-aided NOMA offers a better performance for all $\zeta$ values. Note that the marginal performance improvement also matches with the behavior reported in \cite{Doan_ICC_2018}. Moreover, since a higher value of $\zeta$ makes the popularity profile more skewed towards the first file, the performances of both NOMA and OMA increase with an increase in $\zeta$. This is because the probability of occurrence of case A increases with $\zeta$, which dominates the average performance. The performance improvement offered due to cache-aided NOMA is significant for low $\zeta$, in comparison to OMA.

\begin{figure}[ht]
\begin{center}
 \includegraphics[width=9.2cm, height=6.5cm]{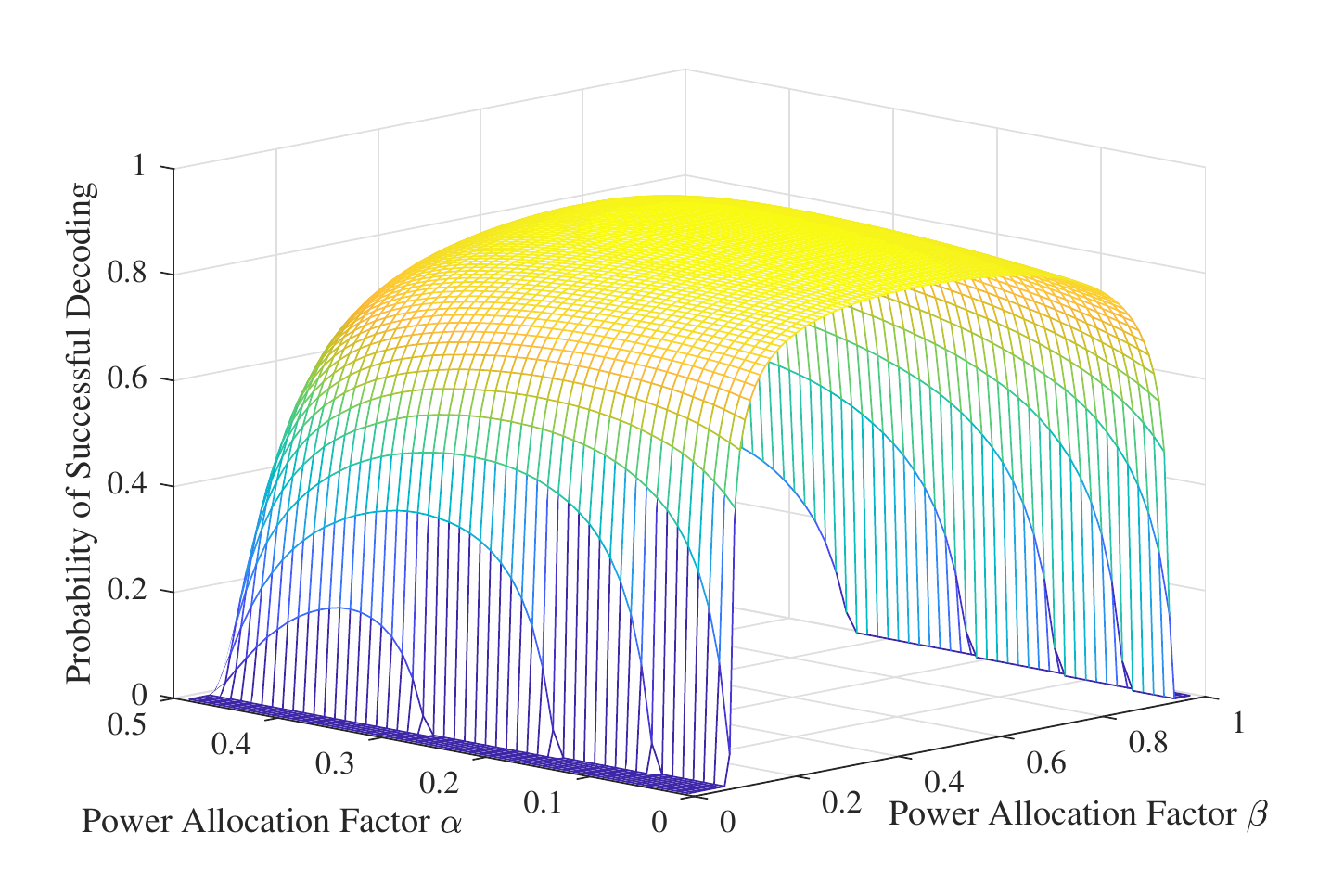}
 \caption{Performance comparison of split file caching-aided NOMA for the cost function given in \eqref{EqnSplitFile_alphalsrp5} with SNR, for different values of power allocation factors $\alpha$ and $\beta$.} \label{FigSplitPsucVsAlpBet}
\end{center}
\end{figure}

Fig.~\ref{FigNvsNN_CachedFiles} compares the performances of cache-aided NOMA and cache-aided OMA, for different number of cached files at $V_1$ and $V_2$. Once again, the achieved  performance   by NOMA outperforms that of OMA and, as expected, increases with an increase in the number of cached files. Also, the achieved performance  due to the double Nakagami$-m$ is outperformed by that of of the conventional Nakagami fading$-m$, as expected. Similarly, Fig. \ref{FigPsucVsMaxCache} shows the impact of cancellation of interference due to caching on the proposed NOMA system. It is seen that in comparison with the conventional NOMA system, the proposed cache-aided NOMA system offers a better performance as the cache size at $V_1$ and $V_2$ increases. Moreover, the performance also improves for lower values of the total files available at the BS, $T$. This is  expected, since a lesser $T$ and a larger cache size improves the probability of caching the requested file in the caching phase. Moreover, when the cache size at $V_1$ and $V_2$ is zero, that is, when $V_1$ and $V_2$ do not cache any files in the caching phase, the performances of conventional NOMA and the caching-aided NOMA are equal.

In Fig.~\ref{FigPsucVsSNR_STCRice}, we investigate the performance of the caching-aided NOMA in a vehicular network with high mobility, where the statistics of the channel is time-varying and has spatial correlation.  Among several spatio-temporally correlated fading models available in the literature \cite{Chen_JSAC_2017}, \cite{Ozdogan_TWC_2019}, \cite{Avazov_AWPL_2017}, we consider the spatio-temporally correlated Rician model, discussed in \cite{Ozdogan_TWC_2019}. In particular, we assume single antenna on the BS and each vehicle, with multiple scatters and the covariance matrix of the observations from each clutter modeled by the approximate Gaussian local scattering model \cite{Ozdogan_TWC_2019}. The model also includes the impairments due to the large scale fading. The Rician parameter for each vehicle is set to unity. The distances between each vehicle and the BS are chosen such that the channel gain is set to a desired value. As shown in Fig.~\ref{FigPsucVsSNR_STCRice}, the probability of successful decoding increases with an increase in channel gain, as expected. In addition to the above result, it is observed that the probability of successful decoding depends on the value Rician parameter, and the degree of correlation over time and space. A detailed study on the impact of time-varying nature of the channel on the performance of the cache-aided NOMA is reserved as a part of the future work.

\begin{figure}[ht]
\begin{center}
 \includegraphics[width=9.2cm, height=6.5cm]{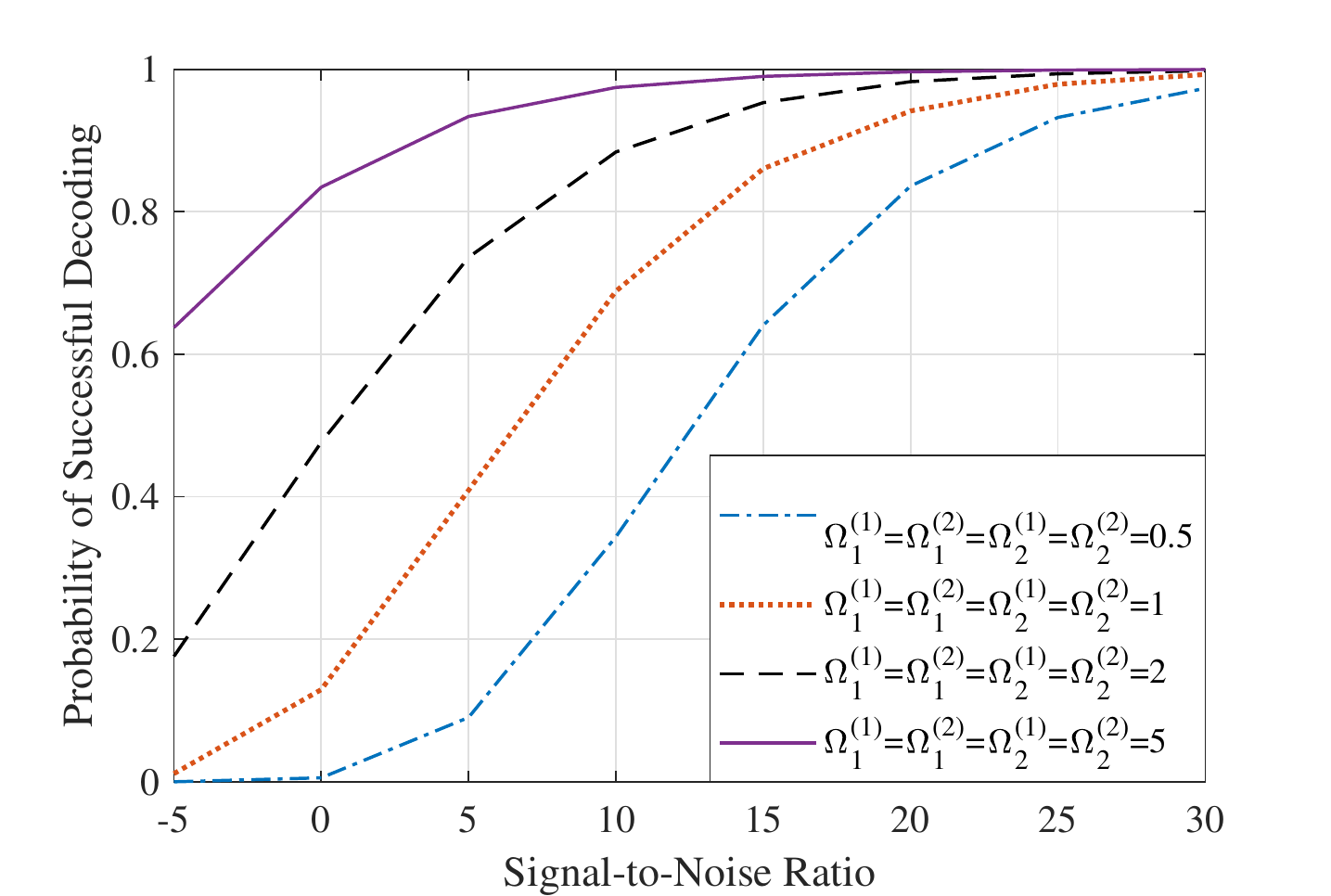}
 \caption{Performance comparison of split file caching-aided NOMA for the cost function given in \eqref{EqnSplitFile_alphalsrp5} with SNR, for different values of the parameters $\omoneone=\omonetwo=\omtwoone=\omtwotwo$, and $\moneone=\monetwo=\mtwoone=\mtwotwo=1$.} \label{FigSplitPsucVsSNRwithOmega}
\end{center}
\end{figure}

\begin{figure}[ht]
\begin{center}
 \includegraphics[width=9.2cm, height=6.5cm]{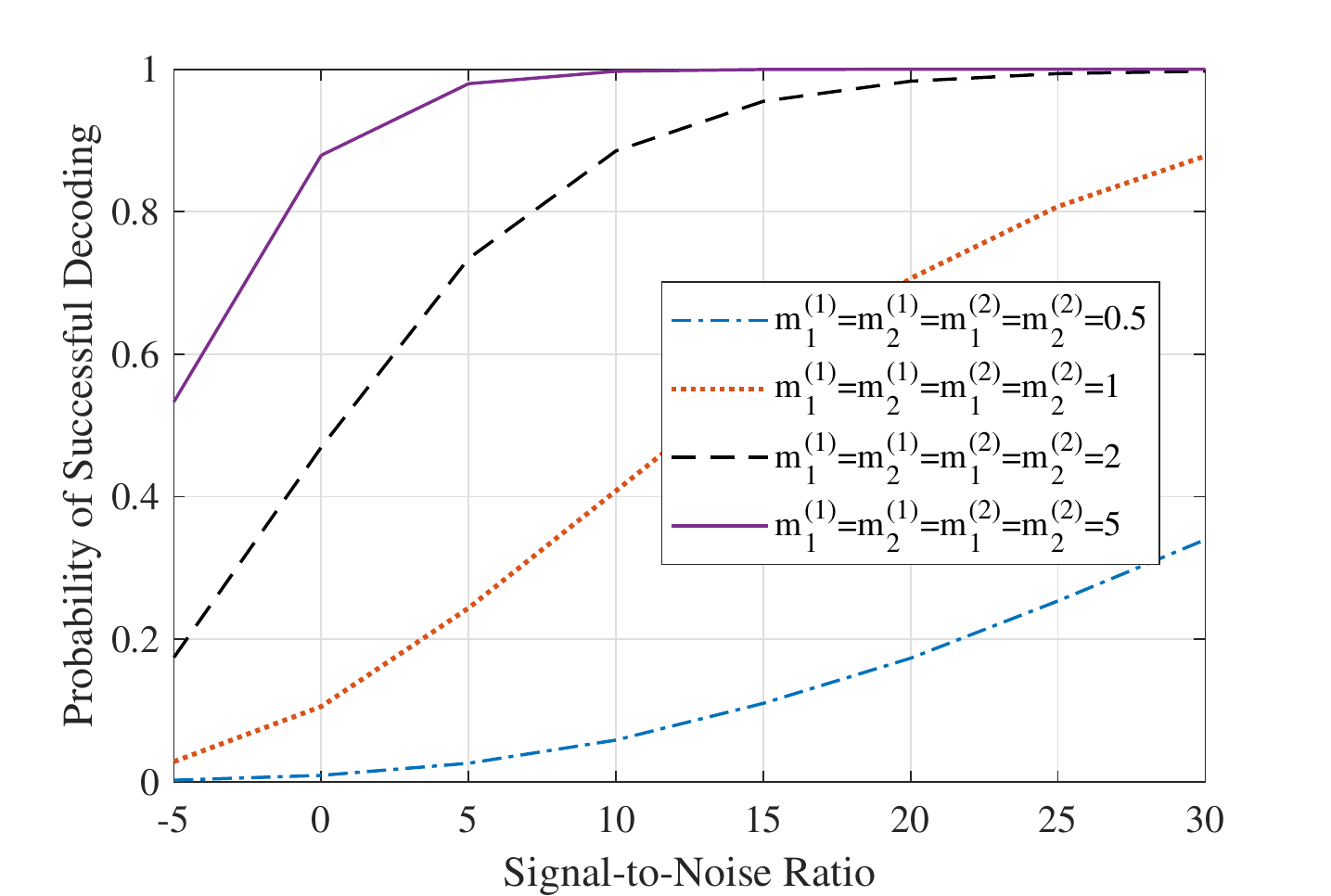}
 \caption{Performance comparison of split file caching-aided NOMA for the cost function given in \eqref{EqnSplitFile_alphalsrp5} with SNR, for different values of the parameters $\moneone=\monetwo=\mtwoone=\mtwotwo$, and $\omoneone=\omonetwo=\omtwoone=\omtwotwo=2$.} \label{FigSplitPsucVsSNRwithm}
\end{center}
\end{figure}

Next, we turn our attention to the case of split file caching discussed in Sec.~\ref{SecSplitFile}. In Fig. \ref{FigSplitPsucVsAlpBet}, we plot the probability of successful decoding given in the cost function of the optimization problem in \eqref{EqnSplitFile_alphalsrp5}, for different values of power allocation parameters $\alpha \in (0,0.5)$ and $\beta \in (0,1)$. The SINR threshold values are set to $\gamma_1^{(1)}=\gamma_2^{(1)}=\gamma_1^{(2)}=\gamma_2^{(2)}=0.25$. It is observed that the cost function is concave in both $\alpha$ and $\beta$. Therefore, the corresponding optimal parameter values $\alpha^*$ and $\beta^*$ can be found by techniques such as the bisection method \cite[Alg.~4.1]{Boyd_CVX}.

Finally, we study the effect of the fading parameters on the proposed  scheme. For simplicity, we again consider the cost function in \eqref{EqnSplitFile_alphalsrp5}, where $\alpha \leq 0.5$. The performance variation of the proposed scheme is   shown in Fig.~\ref{FigSplitPsucVsSNRwithOmega}, for different values of the Nakagami$-m$ parameter $\Omega$. In this scenario, we assume that $\omoneone=\omonetwo=\omtwoone=\omtwotwo=\Omega$, and $\moneone=\monetwo=\mtwoone=\mtwotwo=1$, which coincides with the corresponding Rayleigh case. It is also noted that since the average value of the fading distribution increases with an increase in $\Omega$, the optimal probability of successful decoding also increases. 
A similar trend can be also observed  in Fig.~\ref{FigSplitPsucVsSNRwithm}, where the parameter $\Omega$ is fixed to $\omoneone=\omonetwo=\omtwoone=\omtwotwo=2$ and the parameters $\moneone=\monetwo=\mtwoone=\mtwotwo=m$ are varied. It is evident that, as expected, the performance of the proposed technique increases with an increase in the parameter $m$.

\section{Conclusions and Future Work} \label{SecConc}

This work introduced  cache-aided NOMA in vehicular networks and quantified the achievable performance for different scenarios of interest.
 First, we considered the case of full file caching, where each vehicle caches and requests the whole  files. This was realized for the case of a two-user vehicular network wireless transmission over  double Nakagami$-m$ fading channels. To this effect, we determined  the optimum power allocation that maximizes the overall probability of successful decoding of files at each vehicle. 
 Then, we  considered the case of split file caching, where each file is partitioned into two parts. A joint power allocation optimization problem was formulated, where the power allocation across vehicles and for a single file was determined. In both cases,  it was shown that the associated cost function is concave in the power allocation variables. Furthermore,   a performance gain was exhibited for the  cache-aided NOMA compared to the  conventional NOMA counterpart. 
 
As a part of the future work, several generalizations for the split file caching in Sec.~\ref{SecSplitFile} are possible. For example, one may consider the case of splitting of files with sizes more than two. A more realistic model to consider is the spatio-temporally correlated fading channel \cite{Chen_JSAC_2017}, \cite{Ozdogan_TWC_2019}, given the mobility of both transmitter and receiver. In particular, similar to the Kronecker correlation-based model discussed in \cite{Chen_JSAC_2017}, and the tunnel model considered in \cite{Avazov_AWPL_2017}, a spatio-temporally correlated cascaded Nakagami model can be developed for the setup considered in this work.

\begin{figure*}[ht]
\begin{align}
& \frac{d^2}{d \alpha^2} \, \left[ \alpha ~ G_{2,4}^{2,2}\left(\left. \frac{\xi_1}{\alpha} \right \rvert_{m_1^{(1)}, m_2^{(1)}, 1, 0}^{\hspace{0.25cm} 0, 1 } \right) \right] = \mathcal{K}_{m_2^{(1)}-m_1^{(1)}} \left(2\sqrt{\frac{\xi_1}{\alpha}}\right) \left\{ 2 \left(\frac{\xi_1}{\alpha}\right)^{\frac{m_2^{(1)}+m_1^{(1)}}{2}} - \frac{\xi_1(m_1^{(1)} + m_2^{(1)})}{\alpha} \left(\frac{\xi_1}{\alpha}\right)^{\frac{m_1^{(1)}+m_2^{(1)}}{2}-1} \right\} \nonumber \\
& ~~~~~~~~~~~~~~~~~~~~~~~~~~~~~~~~~~~~~ - \left(\frac{\xi_1}{\alpha}\right)^{\frac{m_1^{(1)}+m_2^{(1)}-1}{2}} \left[-\mathcal{K}_{m_2^{(1)}-m_1^{(1)}-1} \left(2 \sqrt{\left(\frac{\xi_1}{\alpha}\right)}\right) -\mathcal{K}_{m_2^{(1)}-m_1^{(1)}+1} \left(2 \sqrt{\left(\frac{\xi_1}{\alpha}\right)}\right) \right] \tag{46a} \label{SecDer}
\end{align}
\hrulefill
\end{figure*}

\section*{Appendix: Proof of Theorem \ref{ThmMain}} \label{ThmMainProof}
We show that the second derivative of $\psa$ with respect to $\alpha$ is negative in $0 \leq \alpha \leq 1$. It is recalled that $\psvonea$ and $\psvtwoa$ are given by \eqref{CDFForV1}, for $g = \frac{\gamma_1 \sigma_1^2}{\alpha P}$ and $g = \frac{\gamma_2 \sigma_2^2}{(1-\alpha) P}$, respectively. It is also recalled that 
\begin{align}
& \frac{d^u}{dz^u} G_{p,q}^{m,n}\left(\left. \frac{1}{z} \right \rvert_{b_1,\ldots,b_m,b_{m+1},\ldots,b_q} ^{a_1,\ldots,a_n,a_{n+1},\ldots,a_p}\right) = (-1)^u z^u  \nonumber \\
& ~~~~~ G_{p+1,q+1}^{m,n+1} \left(\left. \frac{1}{z} \right \rvert_{b_1,\ldots,b_m,1,b_{m+1},\ldots,b_q}^{1-u,a_1,\ldots,a_n,a_{n+1},\ldots,a_p}  \right), u=1,2,\cdots.
\end{align}
Based on this,  it is noted that the term corresponding to $\psvonea$ can be expressed  as 
\begin{align}
& \dfrac{d}{d\alpha} \left[ \frac{G_{1,3}^{2,1}\left(\left. \frac{m_1^{(1)} m_2^{(1)} \gamma_1 \sigma_1^2}{\Omega_1^{(1)} \Omega_2^{(1)} P \alpha} \right \rvert_{m_1^{(1)}, m_2^{(1)},0}^{\hspace{0.5cm} \text{1} } \right)}{\Gamma(m_1^{(1)}) \Gamma(m_2^{(1)})} \right] = \frac{-\frac{\Omega_1^{(1)} \Omega_2^{(1)} P \alpha}{m_1^{(1)} m_2^{(1)} \gamma_1 \sigma_1^2}}{\Gamma(m_1^{(1)}) \Gamma(m_2^{(1)})} \nonumber \\
& ~~~~~~~~~~~~~~~ \times G_{2,4}^{2,2}\left(\left. \frac{m_1^{(1)} m_2^{(1)} \gamma_1 \sigma_1^2}{\Omega_1^{(1)} \Omega_2^{(1)} P \alpha} \right \rvert_{m_1^{(1)}, m_2^{(1)},1,0}^{\hspace{0.5cm} 0, 1 } \right). \label{EqnwithG}
\end{align}
Now, let us define
\begin{align}
\xi_1 \triangleq \frac{m_1^{(1)} m_2^{(1)} \gamma_1 \sigma_1^2}{\Omega_1^{(1)} \Omega_2^{(1)} P} > 0.
\end{align}
Also, since $m_1 > 0$ and $m_2 > 0$, it follows that 
\begin{align}
& G_{2,4}^{2,2}\left(\left. \frac{\xi_1}{\alpha} \right \rvert_{m_1^{(1)}, m_2^{(1)}, 1, 0}^{\hspace{0.25cm} 0, 1 } \right) = 2 \left(\frac{\xi}{\alpha}\right)^{\frac{m_1^{(1)}+m_2^{(1)}}{2}} \nonumber \\
& ~~~~~~~~~~~~~~~~~~~~~~~~~~~~~~~~~ \times \mathcal{K}_{m_2^{(1)}-m_1^{(1)}} \left(2\sqrt{\frac{\xi_1}{\alpha}}\right).  \label{EqnGSimp}
\end{align}
Evidently, by simplifying \eqref{EqnwithG},  substituting \eqref{EqnGSimp} and finding the second derivative yields a term which can be simplified as given in \eqref{SecDer}, at the top of this page. Similar expressions can be derived for the first and second derivatives of $\psvtwoa$. In addition, utilizing the Bessel function  identity:
\begin{align}
\mathcal{K}_{-\nu}(x) = \mathcal{K}_{\nu}(x), ~~ \nu=0,1,\cdots,
\end{align}
assists  in simplifying $\mathcal{K}_{m_2^{(1)}-m_1^{(1)}}(\cdot)$ irrespective of whether $m_1^{(1)}$ is greater than or lesser than $m_2^{(1)}$. It is also noted that the following inequality holds:
\begin{align}
\mathcal{K}_{\mu}(x) \geq \mathcal{K}_{\nu}(x), ~~\mu,\nu=0,1,\cdots, ~~ x \in \mathbb{R}^+.
\end{align}
Simplifying the first and second derivatives of both $\psvonea$ and $\psvtwoa$ using the above results, it can be shown that the second derivative of $\psa$ is negative, and their respective first derivatives can be shown to have positive and negative values in the limiting cases as $\alpha \rightarrow 0$ and $\alpha \rightarrow 1$, respectively.   Based on this and after some algebraic manipulations,  $\psa$ is shown to be concave in $0 \leq \alpha \leq 1$, which completes the proof. 

\balance
\bibliographystyle{IEEEtran}
\bibliography{IEEEabrv,NOMACachingV2VRef}
\end{document}